\documentclass[runningheads]{llncs}
\usepackage{graphicx}
\usepackage{paperstyle}

\begin{document}
	\title{Automated Termination Analysis of Polynomial Probabilistic Programs \thanks{
			This research was supported by the WWTF ICT19-018 grant ProbInG, the ERC Starting Grant
SYMCAR 639270, the ERC AdG Grant FRAPPANT 787914, and the Austrian FWF project W1255-N23.
	}}
	%
	%
	\author{
		Marcel Moosbrugger\inst{1}\textsuperscript{(\Letter)}\orcidID{0000-0002-2006-3741}\and
		Ezio Bartocci\inst{1}\orcidID{0000-0002-8004-6601} \and \\
		Joost-Pieter Katoen\inst{2}\orcidID{0000-0002-6143-1926} \and
		Laura Kovács\inst{1}\orcidID{0000-0002-8299-2714}
	}
	\authorrunning{M. Moosbrugger et al.}
	%
	\institute{
		TU Wien, Vienna, Austria\\
		\email{marcel.moosbrugger@tuwien.ac.at}
		\and
		RWTH Aachen University, Aachen, Germany
	}
	\maketitle              
	\begin{abstract}
		The termination behavior of probabilistic programs depends on the outcomes of random assignments.
Almost sure termination (AST) is concerned with the question whether a program terminates with probability one on all possible inputs.
Positive almost sure termination (PAST) focuses on termination in a finite expected number of steps.
This paper presents a fully automated approach to the termination analysis of probabilistic while-programs whose guards and expressions are polynomial expressions.
As proving (positive) AST is undecidable in general, existing proof rules typically provide sufficient conditions.
These conditions mostly involve constraints on supermartingales.
We consider four proof rules from the literature and extend these with
generalizations of existing proof rules for (P)AST. 
We automate the resulting set of proof rules by 
effectively computing asymptotic bounds on polynomials over the program variables.
These bounds are used to decide the sufficient conditions -- including the constraints on supermartingales -- of a proof rule.
Our software tool \amber{}
can thus check AST, PAST, as well as their negations for a large class of polynomial probabilistic programs, while carrying out the termination reasoning fully with polynomial witnesses.
Experimental results show the merits of our
generalized proof rules and demonstrate that
\amber{} can handle probabilistic programs that are out of reach for other state-of-the-art tools.

		\keywords{Probabilistic Programming \and Almost sure Termination \and Martingales \and Asymptotic Bounds \and Linear Recurrences}
	\end{abstract}

	\section{Introduction}
\label{section:introduction}
\newacronym{pp}{PP}{probabilistic program}
\newacronym{ast}{AST}{almost surely terminating}
\newacronym{past}{PAST}{positively almost surely terminating}


\paragraph{Classical program termination.}
Termination is a key property in program analysis~\cite{DBLP:journals/cacm/CookPR11}.
The question whether a program terminates on all possible inputs -- the universal halting problem -- is undecidable.
Proof rules based on ranking functions have been developed that impose sufficient conditions implying (non-)termination.
Automated termination checking has given rise to powerful software tools such as AProVE~\cite{DBLP:journals/jar/GieslABEFFHOPSS17} and NaTT~\cite{DBLP:conf/rta/0002KS14} (using term rewriting), and UltimateAutomizer~\cite{DBLP:conf/tacas/HeizmannCDGHLNM18} (using automata theory).
These tools have shown to be able to determine the termination of
several intricate programs. 
The industrial tool Terminator~\cite{cook_terminator_2006} has taken termination proving into practice and is able to prove termination -- or even more general liveness properties -- of e.g., device driver software.
Rather than seeking a single ranking function, it takes a disjunctive termination argument using sets of ranking functions.
Other results include termination proving methods for specific program
classes such as linear and polynomial programs, see, 
e.g.,~\cite{Bradley05,Hark20}.

\begin{figure}[t]
	\begin{subfigure}{0.23\linewidth}
		\centering
		\begin{minipage}{\linewidth}
			\begin{probloop}
				x := 10
				
				\CWhile{$x > 0$}{
					\pvar{x} := $\pvar{x} + 1$ \prob{\nicefrac{1}{2}} $\pvar{x} - 1$
				}
			\end{probloop}
		\end{minipage}
		\caption{}
		\label{1D-random-walk}
	\end{subfigure}
	\begin{subfigure}{0.23\linewidth}
		\centering
		\begin{minipage}{\linewidth}
			\begin{probloop}
				x := 10
				
				\CWhile{$x > 0$}{
					\pvar{x} := $\pvar{x} - 1$ \prob{\nicefrac{1}{2}} $\pvar{x} + 2$
				}
			\end{probloop}
		\end{minipage}
		\caption{}
		\label{not-AST-repulsing}
	\end{subfigure}
	\begin{subfigure}{0.25\linewidth}
		\centering
		\begin{minipage}{\linewidth}
			\begin{probloop}
				x := 0, y := 0
				
				\CWhile{$x^2 + y^2 < 100$} {
					\pvar{x} := $\pvar{x} + 1$ \prob{\nicefrac{1}{2}} $\pvar{x} - 1$
					
					\pvar{y} := $\pvar{y} + x$ \prob{\nicefrac{1}{2}} $\pvar{y} - x$
				}
			\end{probloop}
		\end{minipage}
		\caption{}
		\label{PAST-random}
	\end{subfigure}
	\begin{subfigure}{0.25\linewidth}
		\centering
		\begin{minipage}{\linewidth}
			\begin{probloop}
				x := 10, y := 0
				
				\CWhile{$x > 0$}{
					\pvar{y} := $\pvar{y} + 1$
					
					\pvar{x} := $\pvar{x} + 4y$ \prob{\nicefrac{1}{2}} $\pvar{x} - y^2$
				}
			\end{probloop}
		\end{minipage}
		\caption{}
		\label{PAST-rsm-limit}
	\end{subfigure}
	
	\caption{
		Examples of probabilistic programs in our probabilistic language.
		Program~\ref{1D-random-walk} is a symmetric 1D random walk. The program is \gls{ast} but not \gls{past}.
		Program~\ref{not-AST-repulsing} is not \gls{ast}.
		Programs~\ref{PAST-random} and \ref{PAST-rsm-limit}
		contain dependent variable updates with polynomial
		guards and both programs are \gls{past}.
	}
	\label{fig:motivating}
\end{figure}

\paragraph{Termination of probabilistic program.}
Probabilistic programs extend sequential programs with the ability to draw samples from probability distributions.
They are used e.g. for, encoding randomized algorithms, planning in AI, security mechanisms, and in cognitive science.
In this paper, we consider probabilistic while-programs with discrete probabilistic choices, in the vein of the seminal works~\cite{DBLP:journals/jcss/Kozen81} and~\cite{DBLP:series/mcs/McIverM05}.
Termination of probabilistic programs differs from the classical halting problem in several respects, e.g., probabilistic programs may exhibit diverging runs that have probability mass zero in total.
Such programs do not always terminate, but terminate with probability one -- they \emph{almost surely} terminate.
An example of such a program is given in Figure~\ref{1D-random-walk} where variable $x$ is incremented by $1$ with probability $\nicefrac{1}{2}$, and otherwise decremented with this amount.
This program encodes a one-dimensional (1D) left-bounded random walk starting at position $10$.
Another important difference to classical termination is that the expected number of program steps until termination may be infinite, even if the program almost surely terminates.
Thus, almost sure termination (\gls{ast}) does not imply that the expected number of steps until termination is finite.
Programs that have a finite expected runtime are referred to as \emph{positively almost surely} terminating (\gls{past}).
Figure~\ref{PAST-random} is a sample program that is \gls{past}.
While \gls{past} implies \gls{ast}, the converse does not hold, as
evidenced by Figure~\ref{1D-random-walk}: the program of
Figure~\ref{1D-random-walk} terminates with probability one but needs
infinitely many steps on average to reach $x{=}0$, hence is not \gls{past}.
(The terminology \gls{ast} and \gls{past} was coined in \cite{DBLP:conf/rta/BournezG05} and has its roots in the theory of Markov processes.)

\paragraph{Proof rules for AST and PAST.}
Proving termination of probabilistic programs is hard:
\gls{ast} for a single input is as hard as the universal halting problem,
whereas \gls{past} is even harder~\cite{kaminski_hardness_2015}. 
Termination analysis of probabilistic programs is currently attracting quite some attention.
It is not just of theoretical interest.
For instance, a popular way to analyze probabilistic programs in machine learning is by using some advanced form of simulation.
If, however, a program is not \gls{past}, the simulation may take forever.
In addition, the use of probabilistic programs in safety-critical environments~\cite{arora_netvisa_2013,bistline_bayesian_2015,fremont_scenic_2019} necessitates providing formal guarantees on termination.
Different techniques are considered for probabilistic program termination ranging from probabilistic term rewriting~\cite{DBLP:journals/scp/AvanziniLY20}, sized types~\cite{DBLP:journals/toplas/LagoG19}, and B\"{u}chi automata theory~\cite{ChenH20}, to weakest pre-condition calculi for checking \gls{past}~\cite{DBLP:journals/jacm/KaminskiKMO18}.
A large body of works considers \emph{proof rules} that provide sufficient conditions for proving \gls{ast}, \gls{past}, or their negations.
These rules are based on martingale theory, in particular supermartingales.
They are stochastic processes that can be (phrased in a simplified manner) viewed as the probabilistic analog of ranking functions: the value of a random variable represents the ``value'' of the function at the beginning of a loop iteration.
Successive random variables model the evolution of the program loop.
Being a supermartingale means that the expected value of the random variables at the end of a loop does not exceed its value at the start of the loop.
Constraints on supermartingales form the essential part of
proof rules.
For example, the \gls{ast} proof rule in~\cite{mciver_new_2017}
requires the existence of a supermartingale whose value decreases at
least with a certain amount by at least a certain probability on each
loop iteration. 
Intuitively speaking, the closer the supermartingales comes to zero -- indicating termination -- the more probable it is that it increases more.
The \gls{ast} proof rule in~\cite{mciver_new_2017} is applicable to
prove \gls{ast}
for the program in Figure~\ref{1D-random-walk}; yet, it cannot be used
to prove \gls{past} of  Figures~\ref{PAST-random}-\ref{PAST-rsm-limit}. 
On the other hand, the \gls{past} proof rule
in~\cite{chakarov_probabilistic_2013,ferrer_fioriti_probabilistic_2015}
requires that the expected decrease of the supermartingale 
on each loop iteration is at least some positive constant
$\epsilon$ and on loop termination needs to be at most zero -- very similar to the usual constraint on ranking functions.
While~\cite{chakarov_probabilistic_2013,ferrer_fioriti_probabilistic_2015}
can be used to prove the program in Figure~\ref{PAST-random} to be
\gls{past}, these works cannot be used for
Figure~\ref{1D-random-walk}. They cannot be used 
for proving  Figure~\ref{PAST-rsm-limit} to be \gls{past} either. 
The rule for showing non-\gls{ast}~\cite{chatterjee_stochastic_2017} requires the supermartingale to be repulsing.
This intuitively means that the supermartingale decreases on average with at least $\varepsilon$ and is positive on termination.
Figuratively speaking, it repulses terminating states.
It can be used to prove the program in Figure~\ref{not-AST-repulsing}
to be not \gls{ast}. 
In summary, while existing works for proving \gls{ast}, \gls{past}, and their
negations are generic in nature, they are also restricted for classes of
probabilistic programs. {\it In this paper, we propose
  relaxed versions of existing proof rules for probabilistic
  termination that turn out to treat quite a number of programs that
  could not be proven otherwise
  (Section~\ref{section:relaxations}).}
In particular, (non-)termination of all four programs of
Figure~\ref{fig:motivating} can be proven using our proof rules.

\paragraph{Automated termination checking of AST and PAST.}
Whereas there is a large body of techniques and proof rules, software tool support to automate checking termination of probabilistic programs is still in its infancy.
\emph{This paper presents novel algorithms to automate various proof
  rules for probabilistic programs:} the three aforementioned proof
rules~\cite{chakarov_probabilistic_2013,ferrer_fioriti_probabilistic_2015,mciver_new_2017,chatterjee_stochastic_2017}
and a variant of the non-\gls{ast} proof rule to prove
non-\gls{past}~\cite{chatterjee_stochastic_2017}\footnote{For automation, the proof rule of \cite{mciver_new_2017} is considered for constant decrease and probability functions.}.
We also present relaxed versions of each of the proof rules, going
beyond the state-of-the-art in the termination analysis of
probabilistic programs. 
We focus on so-called Prob-solvable loops,
extending~\cite{bartocci_automatic_2019}. Namely, we define
Prob-solvable loops as probabilistic while-programs whose guards
compare two polynomials (over program variables) and whose body is
a sequence of random assignments with polynomials as right-hand side
such that a variable $x$, say, only depends on variables preceding $x$
in the loop body. While
restrictive, Prob-solvable loops cover a vast set of interesting
probabilistic programs (see Remark~\ref{rem:ProbS}).
An essential property of our programs is that the statistical moments of program variables can be obtained as closed-form formulas~\cite{bartocci_automatic_2019}.
\emph{The key of our algorithmic approach is a procedure for computing
  asymptotic lower, upper and absolute bounds on polynomial
  expressions over program variables in our programs (Section~\ref{section:algorithms}).}
This enables a novel method for automating probabilistic termination
and non-termination proof rules based on (super)martingales,
going
beyond the state-of-the-art in probabilistic termination. Our relaxed
proof rules allow us to fully automate (P)AST analysis by using only polynomial witnesses. Our experiments provide
practical evidence that polynomial witnesses within Prob-solvable loops
are sufficient to certify most examples from the literature and even
beyond (Section~\ref{section:implementation}).

\paragraph{Our termination tool \amber{}.}
We have implemented our algorithmic approach in the publicly available
tool \amber{}.
It exploits asymptotic
bounds over polynomial martingales and uses the tool {\sc mora}~\cite{bartocci_automatic_2019} for
computing the first-order moments of program variables and the
computer algebra system package {\tt diofant}.
It employs over- and under-approximations realized by a simple
static analysis. \emph{\amber{} establishes probabilistic termination in a
fully automated manner} and has the following unique characteristics:
\begin{itemize}
\item it includes the first implementation of the \gls{ast} proof rule of~\cite{mciver_new_2017}, and
\item it is the first tool capable of certifying \gls{ast} for programs that are not \gls{past} and cannot be split into \gls{past} subprograms, and 
\item it is the first tool that brings the various proof rules under a single umbrella: \gls{ast}, \gls{past}, non-\gls{ast} and non-\gls{past}.
\end{itemize}
An experimental evaluation on various benchmarks shows that: (1)
\amber{} is superior to existing tools for automating
\gls{past}~\cite{ngo_bounded_2018} and
\gls{ast}~\cite{chakarov_probabilistic_2013}, (2) the relaxed proof
rules enable proving substantially more programs, and (3) \amber{} is
able to automate the termination checking of intricate probabilistic
programs (within the class of programs considered) that could not be
automatically handled so far
(Section~\ref{section:implementation}). For example, \emph{\amber{} solves
23 termination benchmarks that no other automated approach could so far handle.}

\paragraph{Main contributions.}
To summarize, the main contributions of this paper are:
\begin{enumerate}
\item Relaxed proof rules for (non-)termination, enabling treating a wider class of programs (Section~\ref{section:relaxations}).
\item Efficient algorithms to compute asymptotic bounds on polynomial expressions of program variables (Section~\ref{section:algorithms}).
\item Automation: a realisation of our algorithms in the tool \amber{} (Section~\ref{section:implementation}).
\item Experiments showing the superiority of \amber{} over existing
  tools for proving (P)AST (Section~\ref{section:implementation}).
\end{enumerate}
	
	\ifshort
		\section{Preliminaries}
\label{section:preliminaries}

\newacronym{mdp}{MDP}{Markov Decision Process}
\newacronym{mc}{MC}{Markov chain}
\newacronym{sm}{SM}{supermartingale}

We denote by $\N$ and $\R$ the set of natural and real numbers, respectively. Further, let $\overline{\R}$ denote $\R \union \{ + \infty, - \infty \}$, $\R_0^+$ the non-negative reals and $\R[x_1, \ldots, x_m]$ the polynomial ring in $x_1, \ldots, x_m$ over $\R$.
We write $x := E_{(1)}\ [p_1]\ E_{(2)}\ [p_2] \ldots [p_{m-1}]\ E_{(m)}$ for the probabilistic update of program variable $x$, denoting the execution of $x := E_{(j)}$ with probability $p_{j}$, for $j = 1, \ldots, m-1$, and the execution of $x := E_{(m)}$ with probability $1 - \sum_{j=1}^{m-1} p_{j}$, where $m \in \N$.
We write indices of expressions over program variables in round brackets and use $E_i$ for the stochastic process induced by expression $E$.
This section introduces our programming language extending \emph{Prob-solvable loops}~\cite{bartocci_automatic_2019} and defines the probability space introduced by such programs.
Let $\E$ denote the expectation operator with respect to a probability space.
We assume the reader to be familiar with probability theory~\cite{kemeny_denumerable_1976}.

\subsection{Programming Model: Prob-Solvable Loops}

Prob-solvable loops~\cite{bartocci_automatic_2019} are syntactically restricted probabilistic programs with polynomial expressions over program variables.
The statistical higher-order moments of program variables, like expectation and variance of such loops, can always be computed as functions of the loop counter.
In this paper, we extend {Prob-solvable loops} with polynomial loop guards in order to study their termination behavior, as follows.
\begin{definition}[Prob-solvable loop $\Loop$]\label{def:ProbSolvable}
	A \emph{Prob-solvable loop} $\Loop$ with real-valued variables $x_{(1)}, ..., x_{(m)}$, where $m \in \N$, is a program of the form: $\Init_\Loop\ while\ \Guard_\Loop\ do\ \Update_\Loop\ end$, with
	\begin{itemize}
		\item (Init) $\Init_\Loop$ is a sequence $x_{(1)} := r_{(1)}, ..., x_{(m)} := r_{(m)}$ of $m$ assignments, with $r_{(j)} \in \R$
		
		\item (Guard) $\Guard_\Loop$ is a strict inequality $P > Q$, where $P,\ Q \in \R[x_{(1)}, \ldots, x_{(m)}]$
		
		\item (Update) $\Update_\Loop$ is a sequence of $m$ probabilistic updates of the form
		$$x_{(j)} := a_{(j 1)} x_{(j)} + P_{(j 1)}\ [p_{j 1}]\ a_{(j 2)} x_{(j)} + P_{(j 2)}\ [p_{j 2}]\ ...\ [p_{j (l_j - 1)}]\ a_{(j l_j)} x_{(j)} + P_{(j l_j)},$$
		
		where $a_{(j k)}\in\R_0^+$ are constants, $P_{(j k)}\in\R[x_{(1)},...,x_{(j-1)}]$ are polynomials, $p_{(jk)}\in[0,1]$ and $\sum_{k}\ p_{jk}<1$.
	\end{itemize}
\end{definition}
If $\Loop$ is clear from the context, the subscript $\Loop$ is omitted from $\Init_\Loop$, $\Guard_\Loop$, and $\Update_\Loop$.
Figure~\ref{fig:motivating} gives four example Prob-solvable loops.

\begin{remark}[Prob-solvable expressiveness]\label{rem:ProbS}
	The enforced order of assignments in the loop body of Prob-solvable loops seems restrictive.
	However, many non-trivial probabilistic programs can be
	naturally modeled as succinct Prob-solvable loops.
	These include complex stochastic processes such as 2D random walks and dynamic Bayesian
	networks~\cite{bartocci2020analysis}.
	Almost all existing benchmarks on automated probabilistic termination analysis fall within the scope of Prob-solvable loops (cf. Section~\ref{section:implementation}).
\end{remark}

In the sequel, we consider an arbitrary Prob-solvable loop $\Loop$ and provide all definitions relative to $\Loop$.
The semantics of $\Loop$ is defined next, by associating $\Loop$ with a probability space.

\subsection{Canonical Probability Space}\label{sec:probspace}

A probabilistic program, and thus a Prob-solvable loop, can be semantically described as a probabilistic transition system~\cite{chakarov_probabilistic_2013} or as a probabilistic control flow graph~\cite{chatterjee_stochastic_2017}, which in turn induce an infinite \gls{mc}~\footnote{In fact, \cite{chatterjee_stochastic_2017} consider Markov decision processes, but in absence of non-determinism in Prob-solvable loops, Markov chains suffice for our purpose.}.
An \gls{mc} is associated with a \emph{sequence space}~\cite{kemeny_denumerable_1976}, a special probability space. 
In the sequel, we associate $\Loop$ with the sequence space of its corresponding \gls{mc}, similarly as in~\cite{hark_aiming_2020}.

\begin{definition}[State, Run of $\Loop$]
	The \emph{state} of Prob-solvable loop $\Loop$ over $m$ variables, is a vector $s \in \R^m$.
	Let $s[j]$ or $s[x_{(j)}]$ denote the $j$-th component of $s$ representing the value of the variable $x_{(j)}$ in state $s$.
	A \emph{run} $\vartheta$ of $\Loop$ is an infinite sequence of states.
\end{definition}
Note that any infinite sequence of states is a run.
Infeasible runs will however be assigned measure $0$.
We write $s \models B$ to denote that the logical formula $B$ holds in state $s$.
\begin{definition}[Loop Space of $\Loop$]
	The Prob-solvable loop $\Loop$ induces a canonical filtered probability space $(\Omega^\Loop, \Sigma^\Loop, (\F^\Loop_i)_{i \in \N}, \P^\Loop)$, called \emph{loop space}, where
	\begin{itemize}
		\item the \emph{sample space} $\Omega^\Loop := (\R^m)^\omega$ is the set of all program runs,
		
		\item the \emph{$\sigma$-algebra} $\Sigma^\Loop$ is the smallest $\sigma$-algebra containing all cylinder sets $Cyl(\pi) := \{ \pi \vartheta \mid \vartheta \in (\R^m)^\omega \}$ for all finite prefixes $\pi \in (\R^m)^+$, that is $\Sigma^\Loop := \langle \{ Cyl(\pi) \mid \pi \in (\R^m)^+ \} \rangle_\sigma$,
		
		\item the \emph{filtration} $(\F^\Loop_i)_{i \in \N}$ contains the smallest $\sigma$-algebras containing all cylinder sets for all prefixes of length $i{+}1$, i.e. $\F^\Loop_i := \langle \{ Cyl(\pi) \mid \pi \in (\R^m)^+,\ |\pi| = i{+}1 \} \rangle_\sigma$.
		
		\item the \emph{probability measure} $\P^\Loop$ is defined as $\P^\Loop(Cyl(\pi)) := p(\pi)$, where $p$ is given by
		\begin{equation*}
			p(s) := \mu_\Init(s),\quad
			p(\pi s s') :=
			\begin{cases}
			p(\pi s) \cdot [s' = s], \text{ if } s \models \neg \Guard_{\Loop} \\
			p(\pi s) \cdot \mu_\Update(s, s'), \text{ if } s \models \Guard_{\Loop}.\\
			\end{cases}
		\end{equation*}
		$\mu_\Init(s)$ denotes the probability that, after initialization $\Init_\Loop$, the loop $\Loop$ is in state $s$.
		$\mu_\Update(s, s')$ denotes the probability that, after one loop iteration starting in state $s$, the resulting program state is $s'$.
		$[\ldots]$ represent the Iverson brackets, i.e. $[s'=s]$ is $1$ iff $s'=s$.
	\end{itemize}
\end{definition}
Intuitively, $\P(Cyl(\pi))$ is the probability that prefix $\pi$ is the sequence of the first $|\pi|$ program states when executing $\Loop$.
Moreover, the $\sigma$-algebra $\F_i$ intuitively captures the information about the program run after the loop body $\Update$ has been executed $i$ times.
We note that the effect of the loop body $\Update$ is considered as atomic.

In order to formalize termination properties of a Prob-solvable loop $\Loop$, we define the \emph{looping time} of $\Loop$ to be a random variable in $\Loop$'s loop space.
\begin{definition}[Looping Time of $\Loop$]
	The \emph{looping time} of $\Loop$ is the random variable $T^{\neg \Guard} : \Omega \to \N \union \{ \infty \}$, where $T^{\neg \Guard}(\vartheta) := \inf \{ i \in \N \mid \vartheta_i \models \neg \Guard \}$.
\end{definition}
Intuitively, the looping time $T^{\neg \Guard}$ maps a program run of $\Loop$ to the index of the first state falsifying the loop guard $\Guard$ of $\Loop$ or to $\infty$ if no such state exists.
We now formalize termination properties of $\Loop$ using the looping time $T^{\neg \Guard}$.

\begin{definition}[Termination of $\Loop$]
	The Prob-solvable loop $\Loop$ is \emph{\gls{ast}} if $\P(T^{\neg \Guard} < \infty) = 1$.
	$\Loop$ is \gls{past} if $\E(T^{\neg \Guard}) < \infty$.
\end{definition}

\subsection{Martingales}

While for arbitrary probabilistic programs, answering $\P(T^{\neg \Guard} < \infty)$ and $\E(T^{\neg \Guard} < \infty)$ is undecidable, sufficient conditions for \gls{ast}, \gls{past} and their
negations have been developed~\cite{chakarov_probabilistic_2013,ferrer_fioriti_probabilistic_2015,mciver_new_2017,chatterjee_stochastic_2017}.
These works use (super)martingales which are special stochastic processes.
In this section, we adopt the general setting of martingale theory to a Prob-solvable loop $\Loop$ and then formalize sufficient termination conditions for $\Loop$ in Section~\ref{section:proof-rules}.

\begin{definition}[Stochastic Process of $\Loop$]
	Every arithmetic expression $E$ over the program variables of $\Loop$ induces the stochastic process $(E_i)_{i \in \N}$, $E_i : \Omega \to \R$ with $E_i(\vartheta) := E(\vartheta_i)$.
	For a run $\vartheta$ of $\Loop$, $E_i(\vartheta)$ is the evaluation of $E$ in the $i$-th state of $\vartheta$.
\end{definition}
In the sequel, for a boolean condition $B$ over program variables $x$ of $\Loop$, we write $B_i$ to refer to the result of substituting $x$ by $x_i$ in $B$.

\begin{definition}[Martingales\label{def:martingale}]
	Let $(\Omega, \Sigma, (\F_i)_{i \in \N}, \P)$ be a filtered probability space and $(M_i)_{i \in \N}$ be an integrable stochastic process adapted to $(\F_i)_{i \in \N}$.
	Then $(M_i)_{i \in \N}$ is a \emph{martingale} if $\E(M_{i+1} \mid \F_i) = M_i$ (or equivalently $\E(M_{i+1}{-}M_i \mid \F_i) = 0$).
	Moreover, $(M_i)_{i \in \N}$ is called a \emph{\gls{sm}} if $\E(M_{i+1} \mid \F_i) \leq M_i$ (or equivalently $\E(M_{i+1}{-}M_i \mid \F_i) \leq 0$).
	For an arithmetic expression $E$ over the program variables of $\Loop$, the conditional expected value $\E(E_{i+1} - E_i \mid \F_i)$ is called \emph{the martingale expression of $E$}.
\end{definition}

	\else
		\section{Preliminaries}
\label{section:preliminaries}

\newacronym{mdp}{MDP}{Markov Decision Process}
\newacronym{mc}{MC}{Markov chain}
\newacronym{sm}{SM}{supermartingale}

We denote by $\N$ and $\R$ the set of natural and real numbers, respectively. Further, let $\overline{\R}$ denote $\R \union \{ + \infty, - \infty \}$, $\R_0^+$ the non-negative reals and $\R[x_1, \ldots, x_m]$ the polynomial ring in $x_1, \ldots, x_m$ over $\R$.
We write $x := E_{(1)}\ [p_1]\ E_{(2)}\ [p_2] \ldots [p_{m-1}]\ E_{(m)}$ for the probabilistic update of program variable $x$, denoting the execution of $x := E_{(j)}$ with probability $p_{j}$, for $j = 1, \ldots, m-1$, and the execution of $x := E_{(m)}$ with probability $1 - \sum_{j=1}^{m-1} p_{j}$, where $m \in \N$.
We write indices of expressions over program variables in round brackets and use $E_i$ for the stochastic process induced by expression $E$.
This section introduces our programming language extending \emph{Prob-solvable loops}~\cite{bartocci_automatic_2019} and defines the probability space introduced by such programs.
We assume the reader to be familiar with probability theory~\cite{kemeny_denumerable_1976}.

\subsection{Programming Model: Prob-Solvable Loops}

Prob-solvable loops~\cite{bartocci_automatic_2019} are syntactically restricted probabilistic programs with polynomial expressions over program variables.
The statistical higher-order moments of program variables, like expectation and variance of such loops, can always be computed as functions of the loop counter.
In this paper, we extend {Prob-solvable loops} with polynomial loop guards in order to study their termination behavior, as follows.
\begin{definition}[Prob-solvable loop $\Loop$]\label{def:ProbSolvable}
	A \emph{Prob-solvable loop} $\Loop$ with real-valued variables $x_{(1)}, ..., x_{(m)}$, where $m \in \N$, is a program of the form: $\Init_\Loop\ while\ \Guard_\Loop\ do\ \Update_\Loop\ end$, with
	\begin{itemize}
		\item (Init) $\Init_\Loop$ is a sequence $x_{(1)} := r_{(1)}, ..., x_{(m)} := r_{(m)}$ of $m$ assignments, with $r_{(j)} \in \R$
		
		\item (Guard) $\Guard_\Loop$ is a strict inequality $P > Q$, where $P,\ Q \in \R[x_{(1)}, \ldots, x_{(m)}]$
		
		\item (Update) $\Update_\Loop$ is a sequence of $m$ probabilistic updates of the form
		$$x_{(j)} := a_{(j 1)} x_{(j)} + P_{(j 1)}\ [p_{j 1}]\ a_{(j 2)} x_{(j)} + P_{(j 2)}\ [p_{j 2}]\ ...\ [p_{j (l_j - 1)}]\ a_{(j l_j)} x_{(j)} + P_{(j l_j)},$$
		
		where $a_{(j k)}\in\R_0^+$ are constants, $P_{(j k)}\in\R[x_{(1)},...,x_{(j-1)}]$ are polynomials, $p_{(jk)}\in[0,1]$ and $\sum_{k}\ p_{jk}<1$.
	\end{itemize}
\end{definition}
If $\Loop$ is clear from the context, the subscript $\Loop$ is omitted from $\Init_\Loop$, $\Guard_\Loop$, and $\Update_\Loop$.
Figure~\ref{fig:motivating} gives four example Prob-solvable loops.

\begin{remark}[Prob-solvable expressiveness]\label{rem:ProbS}
	The enforced order of assignments in the loop body of Prob-solvable loops seems restrictive.
	Notwithstanding these syntactic restrictions, many non-trivial probabilistic programs can be
	naturally modeled as succinct Prob-solvable loops.
	These include complex stochastic processes such as 2D random walks and dynamic Bayesian
	networks~\cite{bartocci2020analysis}.
	Almost all existing benchmarks on automated probabilistic termination analysis fall within the scope of Prob-solvable loops (cf. Section~\ref{section:implementation}).
\end{remark}

In the sequel, we consider an arbitrary Prob-solvable loop $\Loop$ and provide all definitions relative to $\Loop$.
The semantics of $\Loop$ is defined next, by associating $\Loop$ with a probability space.

\subsection{Canonical Probability Space}\label{sec:probspace}

A probabilistic program, and thus a Prob-solvable loop, can be semantically described as a probabilistic transition system~\cite{chakarov_probabilistic_2013} or as a probabilistic control flow graph~\cite{chatterjee_stochastic_2017}, which in turn induce an infinite \gls{mc}~\footnote{In fact, \cite{chatterjee_stochastic_2017} consider Markov decision processes, but in absence of non-determinism in Prob-solvable loops, Markov chains suffice for our purpose.}.
An \gls{mc} is associated with a \emph{sequence space}~\cite{kemeny_denumerable_1976}, a special probability space. 
In the sequel, we associate $\Loop$ with the sequence space of its corresponding \gls{mc}, similarly as in~\cite{hark_aiming_2019}.
To this end, we first define the notions \emph{state} and \emph{run} for a Prob-solvable loop.

\begin{definition}[State, Run of $\Loop$]
	The \emph{state} of Prob-solvable loop $\Loop$ over $m$ variables, is a vector $s \in \R^m$.
	Let $s[j]$ or $s[x_{(j)}]$ denote the $j$-th component of $s$ representing the value of the variable $x_{(j)}$ in state $s$.
	A \emph{run} $\vartheta$ of $\Loop$ is an infinite sequence of states.
\end{definition}
Note that any infinite sequence of states is a run.
Infeasible runs will however be assigned measure $0$.
We write $s \models B$ to denote that the logical formula $B$ holds in state $s$.
A probability space $(\Omega, \Sigma, \P)$ consists of a measurable space $(\Omega, \Sigma)$ and a probability measure $\P$ for this space.
First, we define a measurable space for $\Loop$ and later equip it with a probability measure.

\begin{definition}[Loop Space of $\Loop$]
	The Prob-solvable loop $\Loop$ induces a canonical measurable space $(\Omega^\Loop, \Sigma^\Loop)$, called \emph{loop space}, where
    \begin{itemize}
    	\item the sample space $\Omega^\Loop := (\R^m)^\omega$ is the set of all program runs,
		\item the $\sigma$-algebra $\Sigma^\Loop$ is the smallest $\sigma$-algebra containing all cylinder sets $Cyl(\pi) := \{ \pi \vartheta \mid \vartheta \in (\R^m)^\omega \}$ for all finite prefixes $\pi \in (\R^m)^+$, that is $\Sigma^L := \langle \{ Cyl(\pi) \mid \pi \in (\R^m)^+ \} \rangle_\sigma$.
    \end{itemize}
\end{definition}
To turn the loop space of $\Loop$ into a proper probability space, we introduce a probability measure.
To this end, we define the probability $p(\pi)$ of a finite non-empty prefix $\pi$ of a program run.
Let $\mu_\Init(s)$ denote the probability that, after initialization $\Init_\Loop$, the loop $\Loop$ is in state $s$.
Because probabilistic constructs are not allowed in $\Init_\Loop$, $\mu_\Init(s)$ is a Dirac-distribution, such that $\mu_\Init(s) = 1$ for the unique state $s$ defined by $\Init_{\Loop}$ and $\mu_\Init(s') = 0$ for $s' \neq s$.
Moreover, $\mu_\Update(s, s')$ denotes the probability that, after one loop iteration starting in state $s$, the resulting program state is $s'$. 
Note that $\mu_\Init(s)$ and $\mu_\Update(s, s')$ are solely determined by $\Init_\Loop$ and $\Update_\Loop$.
The probability $p(\pi)$ of a finite non-empty prefix $\pi$ of a program run is then defined as
\begin{flalign*}
	p(s) := \mu_\Init(s),\quad
	p(\pi s s') :=
	\begin{cases}
		p(\pi s) \cdot [s' = s], \text{ if } s \models \neg \Guard_{\Loop} \\
		p(\pi s) \cdot \mu_\Update(s, s'), \text{ if } s \models \Guard_{\Loop}\\
	\end{cases}
\end{flalign*}
where $[\ldots]$ denote the Iverson brackets, i.e. $[s'=s]$ is $1$ iff $s'=s$.
Intuitively, $p(\pi)$ is the probability that prefix $\pi$ is the sequence of the first $|\pi|$ program states when executing $\Loop$.
We note that the effect of the loop body $\Update$ is considered as atomic.

\begin{definition}[Loop Measure of $\Loop$]
  The \emph{loop measure} of a Prob-solvable loop $\Loop$ is a canonical probability measure $\P^\Loop : \Sigma^\Loop \to [0,1]$ on the loop space of $\Loop$, with $\P^\Loop(Cyl(\pi)) := p(\pi)$.
\end{definition}
The loop space and the loop measure of $\Loop$ form the probability space $(\Omega^\Loop, \Sigma^\Loop, \P^\Loop)$.

\subsection{Probabilistic Termination}

In order to formalize termination properties of a Prob-solvable loop $\Loop$, we define the \emph{looping time} of $\Loop$ to be a random variable in $\Loop$'s loop space.
A random variable $X$ in a probability space $(\Omega, \Sigma, \P)$ is a ($\Sigma$-)measurable function $X : \Omega \to \overline{\R}$, i.e.\ for every open interval $U \subseteq \overline{\R}$ it holds that $X^{-1}(U) \in \Sigma$.
The expected value of a random variable $X$, denoted by $\E(X)$, is defined as the Lebesgue integral of $X$ over the probability space, i.e. $\E(X) := \int_\Omega X d\P$.
In the special case that $X$ takes only countably many values, we have ${\E(X) = \int_\Omega X d\P = \sum_{r \in X(\Omega)} \P(X = r) \cdot r}$.
We now define the \emph{looping time} of a Prob-solvable loop $\Loop$, as follows. 

\begin{definition}[Looping Time of $\Loop$]
	The \emph{looping time} of $\Loop$ is the random variable $T^{\neg \Guard} : \Omega \to \N \union \{ \infty \}$, where $T^{\neg \Guard}(\vartheta) := \inf \{ i \in \N \mid \vartheta_i \models \neg \Guard \}$.
\end{definition}
Intuitively, the looping time $T^{\neg \Guard}$ maps a program run of $\Loop$ to the index of the first state falsifying the loop guard $\Guard$ of $\Loop$ or to $\infty$ if no such state exists.
We now formalize termination properties of $\Loop$ using the looping time $T^{\neg \Guard}$. 

\begin{definition}[Termination of $\Loop$]
	The Prob-solvable loop $\Loop$ is \emph{\gls{ast}} if $\P(T^{\neg \Guard} < \infty) = 1$.
	$\Loop$ is \gls{past} if $\E(T^{\neg \Guard}) < \infty$.
\end{definition}

\subsection{Filtrations and Martingales}

For a thorough analysis of the hardness of deciding \gls{ast} and \gls{past} we refer to~\cite{kaminski_hardness_2015}.
While for arbitrary probabilistic programs, answering $\P(T^{\neg \Guard} < \infty)$ and $\E(T^{\neg \Guard} < \infty)$ is undecidable, sufficient conditions for \gls{ast}, \gls{past} and their
negations have been developed~\cite{chakarov_probabilistic_2013,ferrer_fioriti_probabilistic_2015,mciver_new_2017,chatterjee_stochastic_2017}.
These works use (super)martingales which are special stochastic processes.
In this section, we adopt the general setting of martingale theory to a Prob-solvable loop $\Loop$ and then formalize sufficient termination conditions for $\Loop$ in Section~\ref{section:proof-rules}.

\begin{definition}[Stochastic Process of $\Loop$]
	A \emph{stochastic process} $(X_i)_{i \in \N}$ is a sequence of random variables.
	Every arithmetic expression $E$ over the program variables of $\Loop$ induces the stochastic process $(E_i)_{i \in \N}$, $E_i : \Omega \to \R$ with $E_i(\vartheta) := E(\vartheta_i)$.
	For a run $\vartheta$ of $\Loop$, $E_i(\vartheta)$ is the evaluation of $E$ in the $i$-th state of $\vartheta$.
\end{definition}
In the sequel, for a boolean condition $B$ over program variables $x$ of $\Loop$, we write $B_i$ to refer to the result of substituting $x$ by $x_i$ in $B$.
In Figure~\ref{1D-random-walk}, the stochastic process $(x_i)_{i \in \N}$ is such that every $x_i$ maps a given program run $\vartheta$ to the value of the variable $x$ in the $i$-th state of $\vartheta$.
Note that the $\sigma$-algebra $\Sigma^\Loop$ contains the cylinder sets for finite program run prefixes of \emph{arbitrary} length.
This does not capture the gradual information gain when executing $\Loop$ iteration by iteration.
In probability theory, \emph{filtrations} are a standard notion to formalize the information available at a specific point in time.

\begin{definition}[Filtration~\cite{kemeny_denumerable_1976}\label{def:fil}]
	For a probability space $(\Omega, \Sigma, \P)$, a \emph{filtration} is a sequence $(\F_i)_{i \in \N}$ such that (1) every $\F_i$ is a sub-$\sigma$-algebra and (2) $\F_i \subseteq \F_{i+1}$.
	Further, $(\Omega, \Sigma, (\F_i)_{i \in \N}, \P)$ is called a \emph{filtered probability space}.
\end{definition}
We adopt filtrations to Prob-solvable loops and enrich the loop space of $\Loop$ to a filtered probability space, as follows.

\begin{definition}[Loop Filtration of $\Loop$\label{def:loopfilt}]
	The \emph{loop filtration} $(\F^\Loop_i)_{i \in \N}$ of $\Sigma^\Loop$ is defined by  $\F^\Loop_i = \langle \{ Cyl(\pi) \mid \pi \in (\R^m)^+,\ |\pi| = i{+}1 \} \rangle_\sigma$. 
	$(\Omega^\Loop, \Sigma^\Loop, (\F^\Loop)_{i \in \N}, \P^\Loop)$ is a \emph{filtered probability space} of $\Loop$.
\end{definition}
Based on Definition~\ref{def:loopfilt}, note that $\F^\Loop_0$ is the smallest $\sigma$-algebra containing the cylinder sets of finite prefixes of program runs of length $1$.
That is, the cylinder sets of finite prefixes of program runs of length greater than or equal to $2$ are not present in $\F^\Loop_0$.
Hence, $\F^\Loop_0$ captures exactly the information available about the program run after executing just the initialization $\Init_\Loop$.
Similarly, $\F^\Loop_i$ captures the information about the program run after the loop body $\Update_\Loop$ has been executed $i$ times.
In Figure~\ref{1D-random-walk}, for example, the event $\{ \vartheta \in \Omega \mid x_i(\vartheta) = r \}$ denoted by $\{ x_i = r \}$ is $\F^\Loop_i$-measurable for every $i \in \N$ and every $r \in \R$, as the value of $x_i$ depends only on information available up to the $i$-th iteration of the loop body of Figure~\ref{1D-random-walk}.
The following definition formalizes this observation.

\begin{definition}[Adapted Process~\cite{kemeny_denumerable_1976}]
	A stochastic process $(X_i)_{i \in \N}$ is said to be \emph{adapted} to a filtration $(\F_i)_{i \in \N}$ if $X_i$ is $\F_i$-measurable for every $i \in \N$.
\end{definition}
It is not hard to argue that, for any arithmetic expression $E$ over the variables of $\Loop$, the induced stochastic process $(E_i)_{i \in \N}$ is adapted to the loop filtration $\F^\Loop_i$ of $\Loop$:
the value of $E_i$ only depends on the information available up to the $i$-th loop iteration of $\Loop$.

The concept of (super)martingales builds upon the notion of \emph{conditional expected values} which is defined as follows.

\begin{definition}[Conditional Expected Value~\cite{kemeny_denumerable_1976}\label{def:CondExpVal}]
	For a probability space $(\Omega, \Sigma, \P)$, an integrable random variable $X$ and a sub-$\sigma$-algebra $\Delta \subseteq \Sigma$, the \emph{expected value of $X$ conditioned on $\Delta$}, $\E(X \mid \Delta)$, is any $\Delta$-measurable function such that for every $D \in \Delta$ we have ${\int_D \E(X \mid \Delta) d\P = \int_D X d\P}$.
	The random variable $\E(X \mid \Delta)$ is almost surely unique.
\end{definition}
We now introduce (super)martingales as special stochastic processes.
In Section~\ref{section:proof-rules} these notions are used to define sufficient conditions for \gls{past}, \gls{ast} and their negations.

\begin{definition}[Martingales\label{def:martingale}]
	Let $(\Omega, \Sigma, (\F_i)_{i \in \N}, \P)$ be a filtered probability space and $(M_i)_{i \in \N}$ be an integrable stochastic process adapted to $(\F_i)_{i \in \N}$.
	Then $(M_i)_{i \in \N}$ is a \emph{martingale} if $\E(M_{i+1} \mid \F_i) = M_i$ (or equivalently $\E(M_{i+1}{-}M_i \mid \F_i) = 0$).
	Moreover, $(M_i)_{i \in \N}$ is called a \emph{\gls{sm}} if $\E(M_{i+1} \mid \F_i) \leq M_i$ (or equivalently $\E(M_{i+1}{-}M_i \mid \F_i) \leq 0$).
	For an arithmetic expression $E$ over the program variables of $\Loop$, the conditional expected value $\E(E_{i+1} - E_i \mid \F_i)$ is called \emph{the martingale expression of $E$}.
\end{definition}
	\fi
	
	\newacronym{rsm}{RSM}{ranking supermartingale}
\newacronym{rsm-rule}{RSM-Rule}{Ranking-Supermartingale-Rule}
\newacronym{sm-rule}{SM-Rule}{Supermartingale-Rule}
\newacronym{r-ast-rule}{R-AST-Rule}{Repulsing-AST-Rule}
\newacronym{r-past-rule}{R-PAST-Rule}{Repulsing-PAST-Rule}

\section{Proof Rules for Probabilistic Termination}
\label{section:proof-rules}

While \gls{ast} and \gls{past} are undecidable in general~\cite{kaminski_hardness_2015},
sufficient conditions, called \emph{proof rules}, for \gls{ast} and \gls{past} have been introduced, see e.g.~\cite{chakarov_probabilistic_2013,ferrer_fioriti_probabilistic_2015,mciver_new_2017,chatterjee_stochastic_2017}.
In this section, we survey four proof rules, adapted to Prob-solvable loops.
In the sequel, a \emph{pure invariant} is a loop invariant in the classical deterministic sense~\cite{Hoare69}.
Based on the probability space corresponding to $\Loop$, a pure invariant holds before and after every iteration of $\Loop$.

\subsection{Positive Almost Sure Termination (PAST)}
\label{section:proof-rules-past}

The proof rule for \gls{past} introduced in \cite{chakarov_probabilistic_2013} relies on the notion of \glspl{rsm}, which is a \gls{sm} that decreases by a fixed positive $\epsilon$ on average at every loop iteration.
Intuitively, \glspl{rsm} resemble ranking functions for deterministic programs, yet for probabilistic programs.

\begin{theorem}[\gls{rsm-rule} \cite{chakarov_probabilistic_2013}, \cite{ferrer_fioriti_probabilistic_2015}]
\label{th:rsm-rule}
	Let $M : \R^m \to \R$ be an expression over the program variables of $\Loop$ and $I$ a pure invariant of $\Loop$.
	Assume the following conditions hold for all $i \in \N$:
	\begin{enumerate}
		\item\label{rsm-rule:term} (Termination)
		$\Guard \land I \implies M > 0$
		\item\label{rsm-rule:rsm} (\gls{rsm} Condition)
		$\Guard_i \land I_i \implies \E(M_{i+1} - M_i \mid \F_i) \leq - \epsilon$, for some $\epsilon > 0$.
	\end{enumerate}
	Then, $\Loop$ is \gls{past}.
	Further, $M$ is called an \emph{$\epsilon$-ranking supermartingale}.
\end{theorem}

\begin{example}
	Consider Figure~\ref{PAST-random}, set $M := 100 - x^2 - y^2$ and $\epsilon := 2$ and let $I$ be $true$.
	Condition~\eqref{rsm-rule:term} of Theorem~\ref{th:rsm-rule} trivially holds.
	Further, $M$ is also an {$\epsilon$-ranking supermartingale}, as $\E(M_{i+1} - M_i \mid \F_i) = 100 - \E(x_{i+1}^2 \mid \F_i) - \E(y_{i+1}^2 \mid \F_i) - 100 + x_i^2 + y_i^2 = -2 - x_i^2 \leq -2$.
	That is because $\E(x_{i+1}^2 \mid \F_i) = x_i^2 + 1$ and $\E(y_{i+1}^2 \mid \F_i) = y_i^2 + x_i^2 + 1$.
	Figure~\ref{PAST-random} is thus proved \gls{past} using the \gls{rsm-rule}.
\end{example}

\subsection{Almost Sure Termination (AST)}\label{section:proof-rules-ast}

Recall that Figure~\ref{1D-random-walk} is \gls{ast} but not \gls{past}, and hence the \gls{rsm}-rule cannot be used for Figure~\ref{1D-random-walk}.
By relaxing the ranking conditions, the proof rule in~\cite{mciver_new_2017} uses general supermartingales to prove \gls{ast} of programs that are not necessarily \gls{past}.

\begin{theorem}[\gls{sm-rule}~\cite{mciver_new_2017}]
	\label{th:sm-rule}
	Let $M : \R^m \to \R_{\geq 0}$ be an expression over the program variables of $\Loop$ and $I$ a pure invariant of $\Loop$.
	Let $p : \R_{\geq 0} \to (0,1]$ (for \emph{probability}) and $d : \R_{\geq 0} \to \R_{> 0}$ (for \emph{decrease}) be antitone (i.e. monotonically decreasing) functions.
	Assume the following conditions hold for all $i \in \N$:
	\begin{enumerate}
		\item (Termination)
		\label{sm-rule:cond:term}
		$\Guard \land I \implies M > 0$
		\item (Decrease)
		\label{sm-rule:cond:dec}
		$\Guard_i  \land I_i \implies \P(M_{i + 1} - M_i \leq -d(M_i) \mid \F_i) \geq p(M_i)$
		\item (\gls{sm} Condition)
		\label{sm-rule:cond:sm}
		$\Guard_i \land I_i \implies \E(M_{i+1} - M_i \mid \F_i) \leq 0$.
	\end{enumerate}
	Then, $\Loop$ is \gls{ast}.
\end{theorem}
Intuitively, the requirement of $d$ and $p$ being antitone forbids that the ``execution~progress'' of $\Loop$ towards termination becomes infinitely small while still being positive.

\begin{example}
	\label{ex:ast:1D} 
	The \gls{sm-rule} can be used to prove \gls{ast} for Figure~\ref{1D-random-walk}. 
	Consider $M := x$, $p := \nicefrac{1}{2}$, $d := 1$ and $I := true$.
	Clearly, $p$ and $d$ are antitone.
	The remaining conditions of Theorem~\ref{th:sm-rule} also hold as \eqref{sm-rule:cond:term} $x > 0 \implies x > 0$;
	\eqref{sm-rule:cond:dec} $x$ decreases by $d$ with probability $p$ in every iteration;
	and \eqref{sm-rule:cond:sm} $\E(M_{i+1} - M_i \mid \F_i) = x_i - x_i \leq 0$.
\end{example}

\subsection{Non-Termination}\label{section:proof-rules-r-ast}

While Theorems~\ref{th:rsm-rule} and \ref{th:sm-rule} can be used for proving \gls{ast} and \gls{past}, respectively, they are not applicable to the analysis of non-terminating Prob-solvable loops.
Two sufficient conditions for certifying the negations of \gls{ast} and \gls{past} have been introduced in~\cite{chatterjee_stochastic_2017} using so-called \emph{repulsing-supermartingales}.
Intuitively, a \emph{repulsing-supermartingale} $M$ on average decreases in every iteration of $\Loop$ and on termination is non-negative.
Figuratively, $M$ repulses terminating states.

\begin{theorem}[\gls{r-ast-rule} \cite{chatterjee_stochastic_2017}]
	\label{th:repsm-ast-rule}
	Let $M : \R^m \to \R$ be an expression over the program variables of $\Loop$ and $I$ a pure invariant of $\Loop$.
	Assume the following conditions hold for all $i \in \N$:
	\begin{enumerate}
		\item (Negative)
		\label{repsm-ast-rule:cond:neg}
		$M_0 < 0$
		\item (Non-Termination)
		\label{repsm-ast-rule:cond:pos}
		$\neg \Guard \land I  \implies M \geq 0$
		\item (\gls{rsm} Condition) 
		\label{repsm-ast-rule:cond:rsm}
		$\Guard_i  \land I_i \implies \E(M_{i+1} - M_i \mid \F_i) \leq - \epsilon$, for some $\epsilon > 0$
		\item ($c$-Bounded Differences)
		\label{repsm-ast-rule:cond:bounded}
		$|M_{i+1} - M_i| < c$, for some $c > 0$.
	\end{enumerate}
	Then, $\Loop$ is \emph{not} \gls{ast}.
	$M$ is called an \emph{$\epsilon$-repulsing supermartingale with $c$-bounded differences}. 
\end{theorem}

\begin{example} \label{ex:nonAST}
	Consider Figure~\ref{not-AST-repulsing} and let $M := -x$, $c := 3$, $\epsilon := \nicefrac{1}{2}$ and $I := true$.
	All four above conditions hold: 
	\eqref{repsm-ast-rule:cond:neg} $-x_0 = -10 < 0$;
	\eqref{repsm-ast-rule:cond:pos} $x \leq 0 \implies -x \geq 0$;
	\eqref{repsm-ast-rule:cond:rsm} $\E(M_{i+1} - M_i \mid \F_i) = -x_i - \nicefrac{1}{2} + x_i = - \nicefrac{1}{2} \leq - \epsilon$;
	and \eqref{repsm-ast-rule:cond:bounded} $|x_i-x_{i+1}| < 3$.
	Thus, Figure~\ref{not-AST-repulsing} is not \gls{ast}.
\end{example}
While Theorem~\ref{th:repsm-ast-rule} can prove programs not to be \gls{ast}, and thus also not \gls{past}, it cannot be used to prove programs not to be \gls{past} when they are \gls{ast}.
For example, Theorem~\ref{th:repsm-ast-rule} cannot be used to prove that Figure~\ref{1D-random-walk} is not \gls{past}.
To address such cases, a variation of the \gls{r-ast-rule}~\cite{chatterjee_stochastic_2017} for certifying programs not to be \gls{past} arises by relaxing the condition $\epsilon > 0$ of the \gls{r-ast-rule} to $\epsilon \geq 0$.
We refer to this variation by \emph{\gls{r-past-rule}}.

\ifshort\else
\begin{example}
	Consider Figure~\ref{1D-random-walk}. We set $M := -x$, $c:=1$ and $\epsilon :=0$.
	Note that $\E(M_{i+1} - M_i \mid \F_i) = -x_i + x_i \leq 0$ and it is easy to see that all four conditions of Theorem~\ref{th:repsm-ast-rule} hold (with $\epsilon \geq 0$).
	Thus, the \gls{r-past-rule} proves that Figure~\ref{1D-random-walk} is not \gls{past}.
\end{example}
\fi

	\section{Relaxed Proof Rules for Probabilistic Termination}
\label{section:relaxations}

While Theorems~\ref{th:rsm-rule}-\ref{th:repsm-ast-rule} provide sufficient conditions proving \gls{past}, \gls{ast} and their negations, the applicability to Prob-solvable loops is somewhat restricted.
For example, the \gls{rsm-rule} cannot be used to prove Figure~\ref{PAST-rsm-limit} to be \gls{past} using the simple expression $M := x$, as explained in detail with Example~\ref{ex:limit-rsm-rule}, but may require more complex witnesses for certifying \gls{past}, complicating automation.
In this section, we relax the conditions of Theorems~\ref{th:rsm-rule}-\ref{th:repsm-ast-rule} by requiring these conditions to only hold ``eventually''.
A property $P(i)$ parameterized by a natural number $i \in \N$ \emph{holds eventually} if there is an $i_0 \in \N$ such that $P(i)$ holds for all $i \geq i_0$.
Our relaxations of probabilistic termination proof rules can intuitively be described as follows:
If $\Loop$, after a fixed number of steps, almost surely reaches a state from which the program is \gls{past} or \gls{ast}, then the program is \gls{past} or \gls{ast}, respectively.
Let us first illustrate the benefits of reasoning with ``eventually'' holding properties for probabilistic termination in the following example.

\begin{figure}
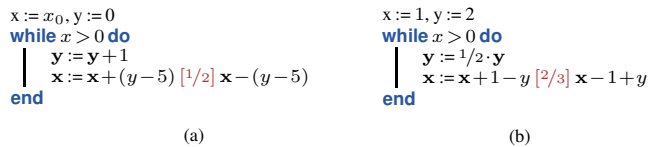

	\centering
	\begin{subfigure}{0.4\linewidth}
		\centering
		\begin{minipage}{\linewidth}
			\begin{probloop}
				x := $x_0$,
				y := 0
				
				\CWhile{$x > 0$}{
					\pvar{y} := $\pvar{y} + 1$
					
					\pvar{x} := $\pvar{x} + (y - 5)$ \prob{\nicefrac{1}{2}} $\pvar{x} - (y - 5)$
				}
			\end{probloop}
		\end{minipage}
		\caption{}
		\label{fig:limit-sm-rule}
	\end{subfigure}
	\begin{subfigure}{0.3\linewidth}
		\centering
		\begin{minipage}{\linewidth}
			\begin{probloop}
				x := $1$, 
				y := 2
				
				\CWhile{$x > 0$} {
					\pvar{y} := $\nicefrac{1}{2} \cdot \pvar{y}$
					
					\pvar{x} := $\pvar{x} + 1 - y$ \prob{\nicefrac{2}{3}} $\pvar{x} - 1 + y$
				}
			\end{probloop}
		\end{minipage}
		\caption{}
		\label{fig:limit-r-ast-rule}
	\end{subfigure}

	\caption{
		Prob-solvable loops which require our relaxed proof rules for termination analysis.
	}
	\label{fig:motivating:relax}
\end{figure}

\begin{example}[Limits of the \gls{rsm-rule} and \gls{sm-rule}]
	\label{ex:limit-rsm-rule}
Consider Figure~\ref{PAST-rsm-limit}. Setting $M := x$, we have the martingale expression $\E(M_{i+1} - M_i \mid \F_i) = -\nicefrac{y_i^2}{2} + y_i + \nicefrac{3}{2} = -\nicefrac{i^2}{2} + i + \nicefrac{3}{2}$.
Since $\E(x_{i+1} - x_i \mid \F_i)$ is non-negative for $i \in \{ 0,1,2,3 \}$, we conclude that $M$ is not an \gls{rsm}. 
However, Figure~\ref{PAST-rsm-limit} either terminates within the first three iterations or, after three loop iterations, is in a state such that the \gls{rsm-rule} is applicable.
Therefore, Figure~\ref{PAST-rsm-limit} is \gls{past} but the \gls{rsm-rule} cannot directly prove using $M := x$.
A similar restriction of the \gls{sm-rule} can be observed for Figure~\ref{fig:limit-sm-rule}. By considering $M := x$, we derive the martingale expression $\E(x_{i+1} - x_i \mid \F_i) = 0$, implying that $M$ is a martingale for Figure~\ref{fig:limit-sm-rule}. 
However, the decrease function $d$ for the \gls{sm-rule} cannot be
defined because, for example, in the fifth loop iteration of
Figure~\ref{fig:limit-sm-rule}, there is no progress as $x$ is
almost surely updated with its previous value. 
However, after the fifth iteration of Figure~\ref{fig:limit-sm-rule}, $x$ always decreases by at least $1$ with probability $\nicefrac{1}{2}$ and all conditions of the \gls{sm-rule} are satisfied.
Thus, Figure~\ref{fig:limit-sm-rule} either terminates within the first five iterations or reaches a state from which it terminates almost surely.
Consequently, Figure~\ref{fig:limit-sm-rule} is \gls{ast} but the \gls{sm-rule} cannot directly prove it using $M := x$.
\end{example}

We therefore relax the \gls{rsm-rule} and \gls{sm-rule} of Theorems~\ref{th:rsm-rule} and \ref{th:sm-rule} as follows.

\begin{theorem}[Relaxed Termination Proof Rules]\label{relax-term}
	For the \gls{rsm-rule} to certify \gls{past} of $\Loop$, it is sufficient that conditions~\eqref{rsm-rule:term}-\eqref{rsm-rule:rsm} of Theorem \ref{th:rsm-rule} hold eventually (instead of for all $i \in \N$).
	Similarly, for the \gls{sm-rule} to certify \gls{ast} of $\Loop$, it is sufficient that conditions~\eqref{sm-rule:cond:term}-\eqref{sm-rule:cond:sm} of Theorem \ref{th:sm-rule} hold eventually.
\end{theorem}
\begin{proof}
	We prove the relaxation of the \gls{rsm-rule}.
	The proof of the relaxed \gls{sm-rule} is analogous.
	Let $\Loop:= \Init\ while\ \Guard\ do\ \Update\ end$ be as in Definition~\ref{def:ProbSolvable}.
	Assume $\Loop$ satisfies the conditions~\eqref{rsm-rule:term}-\eqref{rsm-rule:rsm} of Theorem \ref{th:rsm-rule} after some $i_0 \in \N$.
	We construct the following probabilistic program $\mathcal{P}$, where $i$ is a new variable not appearing in $\Loop$:
	\begin{equation}\label{relaxed:PP}
		\begin{array}{l}
			\Init; i:=0\\
			while\ i<i_0\ do\ \Update; i:=i+1\ end\\
			while\ \Guard\ do\ \Update\ end
		\end{array}
	\end{equation}
	We first argue that if $\mathcal{P}$ is \gls{past}, then so is $\Loop$.
	Assume $\mathcal{P}$ to be \gls{past}.
	Then, the looping time of $\Loop$ is either bounded by $i_0$ or it is \gls{past}, by the definition of $\mathcal{P}$.
	In both cases, $\Loop$ is \gls{past}.
	Finally, observe that $\mathcal{P}$ is \gls{past} if and only if its second while-loop is \gls{past}.
	However, the second while-loop of $\mathcal{P}$ can be certified to be \gls{past} using the \gls{rsm-rule} and additionally using $i \geq i_0$ as an invariant.\qed
\end{proof}

\begin{remark}
	The central point of our proof rule relaxations is that they allow for simpler witnesses.	
	While for Example~\ref{ex:limit-rsm-rule} it can be checked that $M := x + 2^{y + 5}$ is an \gls{rsm}, the example illustrates that the relaxed proof rule allows for a much simpler \gls{past} witness (linear instead of exponential).
	This simplicity is key for automation.
\end{remark}

Similar to Theorem~\ref{relax-term}, we relax the \gls{r-ast-rule} and the \gls{r-past-rule}. 
However, compared to Theorem~\ref{relax-term}, it is not enough for a non-termination proof rule to certify non-\gls{ast} from some state onward, because $\Loop$ may never reach this state as it might terminate earlier.
Therefore, a necessary assumption when relaxing non-termination proof rules comes with ensuring that $\Loop$ has a positive probability of reaching the state after which a proof rule witnesses non-termination.
This is illustrated in the following example .

\begin{example}[Limits of the \gls{r-ast-rule}]
	\label{ex:limit-r-ast-rule}
	Consider Figure~\ref{fig:limit-r-ast-rule} and set $M := -x$.
	As a result, we get $\E(M_{i+1} - M_i \mid \F_i) = \nicefrac{y_i}{6} - \nicefrac{1}{3} = \nicefrac{2^{-i}}{3} - \nicefrac{1}{3}$.
	Thus, $\E(M_{i+1} - M_i \mid \F_i) = 0$ for $i = 0$, implying that $M$ cannot be an $\epsilon$-repulsing supermartingale with $\epsilon > 0$ for all $i \in \N$.
	However, after the first iteration of $\Loop$, $M$ satisfies all requirements of the \gls{r-ast-rule}.
	Moreover, $\Loop$ always reaches the second iteration because in the first iteration $x$ almost surely does not change.
	From this follows that Figure~\ref{fig:limit-r-ast-rule} is not \gls{ast}.
\end{example}

The following theorem formalizes the observation of Example~\ref{ex:limit-r-ast-rule} relaxing the \gls{r-ast-rule} and \gls{r-past-rule} of Theorem~\ref{th:repsm-ast-rule}.

\begin{theorem}[Relaxed Non-Termination Proof Rules for]
	\label{relax-nonterm}
	For the \gls{r-ast-rule} to certify non-\gls{ast} for $\Loop$ (Theorem \ref{th:repsm-ast-rule}), as well as for the \gls{r-past-rule} to certify non-\gls{past} for $\Loop$ (Theorem \ref{th:repsm-ast-rule}), if $\P(M_{i_0} < 0) > 0$ for some $i_0 \geq 0$, it suffices that conditions \eqref{repsm-ast-rule:cond:pos}-\eqref{repsm-ast-rule:cond:bounded} hold for all $i \geq i_0$ (instead of for all $i \in \N$).
\end{theorem}

\ifshort\else
\begin{proof}
	We prove the relaxation of the \gls{r-ast-rule}.
	The proof for the \gls{r-past-rule} is analogous. 
	Let $\Loop:= \Init\ while\ \Guard\ do\ \Update\ end$ be as in Definition~\ref{def:ProbSolvable}.
	Assume $\Loop$ satisfies conditions \eqref{repsm-ast-rule:cond:pos}-\eqref{repsm-ast-rule:cond:bounded} of the \gls{r-ast-rule} for all $i \geq i_0$ for some fixed $i_0 \in \N$.
	Moreover, assume $\P(M_{i_0} < 0) > 0$.

	We construct again a probabilistic program $\mathcal{P}$ as in~\eqref{relaxed:PP}.
	Observe that for the second while-loop of $\mathcal{P}$, we have $i \geq i_0$.
	By assumption, the second while-loop of $\mathcal{P}$ satisfies conditions \eqref{repsm-ast-rule:cond:pos}-\eqref{repsm-ast-rule:cond:bounded} of the \gls{r-ast-rule}.
	By the \gls{r-ast-rule}, we conclude $\mathcal{P}$ being not \gls{ast}, if there is a $Cyl(\pi) \in \F^{\mathcal{P}}_{i_0}$, such that $\P^{\mathcal{P}}(Cyl(\pi)) > 0$ and $M_{i_0}(\vartheta) < 0$ for all $\vartheta \in Cyl(\pi)$.
	
	By the definition of $\mathcal{P}$, it then follows for $\Loop$ that if there is a $Cyl(\pi) \in \F^\Loop_{i_0}$, such that $\P^\Loop(Cyl(\pi)) > 0$ and $M_{i_0}(\vartheta) < 0$ for all $\vartheta \in Cyl(\pi)$, then $\Loop$ is not \gls{ast}.
	As $\P^\Loop(M_{i_0} < 0) > 0$, we conclude that such a $Cyl(\pi)$ exists and derive that $\Loop$ is not \gls{ast}.\qed
\end{proof}
Note that for a repulsing supermartingale $M$, the condition $\P(M_{i_0} < 0) > 0$ implies that there is a positive probability of reaching iteration $i_0$, because $M$ would have to be almost surely non-negative upon termination.
\fi

\ifshort
The proof of Theorem~\ref{relax-nonterm} is similar to the one of
Theorem~\ref{relax-term} and available in~\cite{moosbrugger2020automated}.
\fi
In what follows, whenever we write \gls{rsm-rule}, \gls{sm-rule}, \gls{r-ast-rule} or \gls{r-past-rule} we refer to our relaxed versions of the proof rules.
	
	\section{Algorithmic Termination Analysis through Asymptotic Bounds}
\label{section:algorithms}

The \emph{two major challenges when automating reasoning} with the
proof rules of Sections~\ref{section:proof-rules} and \ref{section:relaxations} are
(i) constructing expressions $M$ over the program variables and
(ii) proving inequalities involving ${\E(M_{i+1} - M_i \mid \F_i)}$.
In this section, we address these two challenges for Prob-solvable loops.
For the loop guard $\Guard_\Loop = P > Q$, let $G_\Loop$ denote the polynomial $P - Q$.
As before, if $\Loop$ is clear from the context, we omit the subscript $\Loop$.
It holds that $G > 0$ is equivalent to $\Guard$.

\paragraph{(i) Constructing (super)martingales $M$:}

For a Prob-solvable loop $\Loop$, the polynomial $G$ is a natural candidate for the expression $M$ in termination proof rules (\gls{rsm-rule}, \gls{sm-rule}) and $-G$ in the non-termination proof rules (\gls{r-ast-rule}, \gls{r-past-rule}).
Hence, we construct potential (super)martingales $M$ by setting $M := G$ for the \gls{rsm-rule} and the \gls{sm-rule}, and $M := -G$ for the \gls{r-ast-rule} and the \gls{r-past-rule}.
The property $\Guard \implies G > 0$, a condition of the \gls{rsm-rule} and the \gls{sm-rule}, trivially holds.
Moreover, for the \gls{r-ast-rule} and \gls{r-past-rule} the condition $\neg \Guard \implies -G \geq 0$ is satisfied.
The remaining conditions of the proof rules are:
\begin{itemize}
	\item \gls{rsm-rule}:
		(a) $\Guard_i \implies \E(G_{i+1} - G_i \mid \F_i) \leq - \epsilon$ for some $\epsilon > 0$
	\item \gls{sm-rule}:
		(a) $\Guard_i \implies \E(G_{i+1} - G_i \mid \F_i) \leq 0$ and
		(b) $\Guard_i \implies \P(G_{i + 1} - G_i \leq - d \mid \F_i) \geq p$ for some $p \in (0, 1]$ and $d \in \R^+$
		(for the purpose of efficient automation, we restrict the functions $d(r)$ and $p(r)$ to be constant)
	\item \gls{r-ast-rule}:
		(a) $\Guard_i \implies \E(-G_{i+1} {+} G_i \mid \F_i) \leq - \epsilon$ for some $\epsilon > 0$ and
		(b) $|G_{i+1} - G_i| \leq c$, for some $c > 0$.
\end{itemize}
All these conditions express bounds over $G_i$.
Choosing $G$ as the potential witness may seem simplistic.
However, Example~\ref{ex:limit-rsm-rule} already illustrated how our relaxed proof rules can mitigate the need for more complex witnesses (even exponential ones).
\emph{The computational effort in our approach does not lie in synthesizing a complex witness but in constructing asymptotic bounds for the loop guard.}
Our approach can therefore be seen as complementary to approaches synthesizing more complex witnesses~\cite{chakarov_probabilistic_2013,chatterjee_termination_2016,chatterjee_stochastic_2017}.
The martingale expression $\E(G_{i+1} - G_i \mid \F_i)$ is an expression over program variables, whereas $G_{i+1} - G_i$ cannot be interpreted as a single expression but through a distribution of expressions.

\begin{definition}[One-step Distribution]\sloppy
	For expression $H$ over the program variables of Prob-solvable loop $\Loop$, let the \emph{one-step distribution} $\Update_\Loop^H$ be defined by ${E \mapsto \P(H_{i+1} = E \mid \F_i)}$ with support set $\supp(\Update_\Loop^H) := \{ B \mid \Update_\Loop^H(B) > 0 \}$.
	We refer to expressions $B \in \supp(\Update_\Loop^H)$ by \emph{branches of $H$}.
\end{definition}
The notation $\Update_\Loop^H$ is chosen to suggest that the loop body $\Update_\Loop$ is ``applied'' to the expression $H$, leading to a distribution over expressions.
Intuitively, the support $\supp(\Update_\Loop^H)$ of an expression $H$ contains all possible updates of $H$ after executing a single iteration of $\Update_\Loop$.

\ifshort\else
\begin{example}[One-step Distribution]
	Consider the following Prob-solvable loop:
	
	\begin{minipage}{\linewidth}
		\begin{probloop}
			x := 1, y := 1
			
			\CWhile{$x > 0$}{
				\pvar{y} := $\pvar{y} + 1$ \prob{\nicefrac{1}{2}} $\pvar{y} + 2$
				
				\pvar{x} := $\pvar{x} + y$ \prob{\nicefrac{1}{3}} $\pvar{x} - y$
			}
		\end{probloop}
		\vspace{0.5em}
	\end{minipage}
	For the expression $H := x^2$, the one-step distribution $\Update_\Loop^H$ is as follows:
	
	\begin{minipage}{\linewidth}
		\begin{minipage}{0.4\linewidth}
			\begin{table}[H]
				\bgroup
				\tiny
				\def\arraystretch{1.5}
				\begin{tabular}{cc}
					\toprule
					Expression E & $\Update_\Loop^H(E)$ \\
					\midrule
					$x_i^2 + 2x_iy_i + 2x_i + y_i^2 + 2y_i + 1$ & $\nicefrac{1}{6}$ \\
					\hdashline
					$x_i^2 + 2x_iy_i + 4x_i + y_i^2 + 4y_i + 4$ & $\nicefrac{1}{6}$ \\
					\hdashline
					$x_i^2 - 2x_iy_i - 2x_i + y_i^2 + 2y_i + 1$ & $\nicefrac{1}{3}$ \\
					\hdashline
					$x_i^2 - 2x_iy_i - 4x_i + y_i^2 + 4y_i + 4$ & $\nicefrac{1}{3}$ \\
					\hdashline
					Any other $E$ & $0$ \\
					\bottomrule
				\end{tabular}
				\egroup
			\end{table}
		\end{minipage}
		\begin{minipage}{0.6\linewidth}
			\vspace{2em}
			The first entry in the table can be derived like:
			\begin{flalign*}
				x_{i+1}^2 &=
				(x_i + y_{i+1})^2 = x_i^2 + 2 x_i y_{i+1} + y_{i+1}^2 \\
				& \text{(with probability $\nicefrac{1}{3}$)} \\
				&= x_i^2 + 2 x_i (y_i + 1) + (y_i + 1)^2 \\
				& \text{(with probability $\nicefrac{1}{2} \cdot \nicefrac{1}{3}$)} \\
				&= x_i^2 + 2 x_i y_i + 2 x_i + y_i^2 + 2 y_i + 1 \\
				& \text{(with probability $\nicefrac{1}{6}$)}
			\end{flalign*}
		\end{minipage}
	\end{minipage}
\end{example}
\fi

\paragraph{(ii) Proving inequalities involving $\E(M_{i+1} - M_i \mid \F_i)$:}
To automate the termination analysis of $\Loop$ with the proof rules from Section~\ref{section:proof-rules}, we need to compute bounds for the expression ${\E(G_{i+1} - G_i \mid \F_i)}$ as well as for the branches of $G$.
In addition, our relaxed proof rules from Section~\ref{section:relaxations} only need asymptotic bounds, i.e. bounds which hold eventually.
In Section~\ref{subsec:bounds}, we propose Algorithm~\ref{alg:monom-bound} for computing \emph{asymptotic lower and upper bounds} for any polynomial expression over program variables of $\Loop$.
Our procedure allows us to derive bounds for $\E(G_{i+1} - G_i \mid \F_i)$ and the branches of $G$.
Before formalizing our method, let us first illustrate how reasoning with asymptotic bounds helps to apply termination proof rules to $\Loop$.

\begin{example}[Asymptotic Bounds for the \gls{rsm-rule}]
	\label{ex:bounds:rsm}
	Consider the following program:
	
	\begin{minipage}{\linewidth}
		\begin{probloop}
			x := 1, y := 0
			
			\CWhile{$x < 100$}{
				\pvar{y} := $\pvar{y} + 1$
				
				\pvar{x} := $2\pvar{x} + y^2$ \prob{\nicefrac{1}{2}} $\nicefrac{1}{2} \cdot \pvar{x}$
			}
		\end{probloop}
		\vspace{0.5em}
	\end{minipage}
	Observe $y_i = i$.
	The martingale expression for $G = 100 - x$ is
	$\E(G_{i+1} - G_i \mid \F_i) = \nicefrac{1}{2}(100 - 2 x_i - (i+1)^2) + \nicefrac{1}{2}(100 - \nicefrac{x_i}{2}) - (100 -x_i) = -\nicefrac{x_i}{4} - \nicefrac{i^2}{2} - i - \nicefrac{1}{2}.$
	Note that if the term $- \nicefrac{x_i}{4}$ would not be present in $\E(G_{i+1} - G_i \mid \F_i)$, we could certify the program to be \gls{past} using the \gls{rsm-rule} because $-\nicefrac{i^2}{2} - i - \nicefrac{1}{2} \leq - \nicefrac{1}{2}$ for all $i \geq 0$.
	However, by taking a closer look at the variable $x$, we observe that it is \emph{eventually} and almost surely lower bounded by the function $\alpha \cdot 2^{-i}$ for some $\alpha \in \R^+$.
	Therefore, \emph{eventually} $- \nicefrac{x_i}{4} \leq - \beta \cdot 2^{-i}$ for some $\beta \in \R^+$.
	Thus, \emph{eventually} $\E(G_{i+1} - G_i \mid \F_i) \leq -\gamma \cdot i^2$ for some $\gamma \in \R^+$.
	By our \gls{rsm-rule}, the program is \gls{past}.
	
	Now, the question arises how the asymptotic lower bound $\alpha \cdot 2^{-i}$ for $x$ can be computed automatically.
	In every iteration, $x$ is either updated with $2x + y^2$ or $\nicefrac{1}{2} \cdot x$.
	Considering the updates as recurrences, we have the inhomogeneous parts $y^2$ and $0$.
	Asymptotic lower bounds for these parts are $i^2$ and $0$, respectively, where $0$ is the ``asymptotically smallest one``.
	Taking $0$ as the inhomogeneous part, we construct two recurrences: (1) $l_0 = \alpha,\ l_{i+1} = 2 l_i + 0$ and (2) $l_0 = \alpha,\ l_{i+1} = \nicefrac{1}{2} \cdot l_i + 0$, for some $\alpha \in \R^+$.
	Solutions to these recurrences are $\alpha \cdot 2^i$ and $\alpha \cdot 2^{-i}$, where the last one is the desired lower bound because it is ``asymptotically smaller``.
	We will formalize this idea of computing asymptotic bounds in Algorithm~\ref{alg:monom-bound}.
\end{example}

We next present our method for computing asymptotic bounds over martingale expressions in Sections~\ref{subsec:monomials}-\ref{subsec:bounds}.
Based on these asymptotic bounds, in Section~\ref{subsec:algorithms-proof-rules} we introduce
algorithmic approaches for our proof rules from Section~\ref{section:relaxations}, solving our aforementioned challenges (i)-(ii) in a fully automated manner (Section~\ref{subsec:rule-out-rules}).

\subsection{Prob-solvable Loops and Monomials}
\label{subsec:monomials}

Algorithm~\ref{alg:monom-bound} computes asymptotic bounds on monomials over program variables in a recursive manner.
To ensure termination of Algorithm~\ref{alg:monom-bound}, it is important that there are no circular dependencies among monomials.
By the definition of Prob-solvable loops, this indeed holds for program variables (monomials of order 1).
Every Prob-solvable loop $\Loop$ comes with an ordering on its variables and every variable is restricted to only depend linearly on itself and polynomially on previous variables.
Acyclic dependencies naturally extend from single variables to monomials.

\begin{definition}[Monomial Ordering]
	Let $\Loop$ be a Prob-solvable loop with variables $x_{(1)}, ..., x_{(m)}$.
	Let $y_1 = \prod_{j=1}^{m} x_{(j)}^{p_j}$ and $y_2 = \prod_{j=1}^{m} x_{(j)}^{q_j}$, where $p_j, q_j \in \N$, be two monomials over the program variables.
	The \emph{order $\preceq$ on monomials} over the program variables of $\Loop$ is defined by
	$y_1 \preceq y_2 \iff (p_m, ..., p_1) \leq_{lex} (q_m, ..., q_1)$,
	where $\leq_{lex}$ is the lexicographic order on $\N^m$.
	The order $\preceq$ is total because $\leq_{lex}$ is total.
	With $y_1 \prec y_2$ we denote $y_1 \preceq y_2 \land y_1 \neq y_2$.
\end{definition}

\ifshort\else
\begin{example}[Monomials]
	Let $\Loop$ be a Prob-solvable loop with variables $x_{(1)}, ..., x_{(m)}$.
	The following statements hold for the monomial order $\preceq$:
	
	$1 \prec x_{(1)} \prec x_{(2)} \prec ... \prec x_{(m-1)} \prec x_{(m)}$, $x_{(1)}^k \prec x_{(2)}$ for any $k \in \N$
	
	$x_{(1)}^2 \prec x_{(1)}^3$ and $x_{(3)}^4 x_{(2)}^{100} x_{(1)}^{99} \prec x_{(3)}^5 x_{(2)}^2 x_{(1)}^3$.
\end{example}
\fi

\noindent
To prove acyclic dependencies for monomials we exploit the following fact.
\begin{lemma}
	\label{lemma:monomial-order}
	Let $y_1, y_2, z_1, z_2$ be monomials.
	If $y_1 \preceq z_1$ and $y_2 \preceq z_2$ then $y_1 \cdot y_2 \preceq z_1 \cdot z_2$.
\end{lemma}

\noindent
By structural induction over monomials and Lemma~\ref{lemma:monomial-order}, we establish:

\begin{lemma}[Monomial Acyclic Dependency]
	\label{lemma:monomial-acylic}
	Let $x$ be a monomial over the program variables of $\Loop$.
	For every branch $B \in \supp(\Update_\Loop^x)$ and monomial $y$ in $B$, $y \preceq x$ holds.
\end{lemma}
\ifshort\else
\begin{proof}
	We use structural induction over monomials.
	The base case for which $x$ is a single variable holds by the definition of $\Loop$ being a Prob-solvable loop.
	Let $x := s \cdot t$ where $s$ and $t$ are monomials over the variables of $\Loop$ and
	\begin{itemize}
		\item for every $B_s \in \supp(\Update_\Loop^s)$ and every monomial $u$ in $B_s$ it holds that $u \preceq s$,
		\item for every $B_t \in \supp(\Update_\Loop^t)$ and every monomial $w$ in $B_t$ it holds that $w \preceq t$,
	\end{itemize}
	Let $B \in \supp(\Update_\Loop^x)$ be an arbitrary branch of $x$.
	By definition of $\Update_\Loop^x$, we get $B = B_s \cdot B_t$, where $B_s$ is a branch of $s$ and $B_t$ is a branch of $t$.
	Note that $B_s$ and $B_t$ are polynomials over program variables or equivalently linear combinations of monomials.
	Therefore, for every monomial $y$ in $B$ we have $y = u \cdot w$ where $u$ is a monomial in $B_s$ and $w$ a monomial in $B_t$.
	By the induction hypothesis, $u \preceq s$ and $w \preceq t$.
	Using Lemma~\ref{lemma:monomial-order}, we get $u \cdot w \preceq s \cdot t$ which means $y \preceq x$.\qed
\end{proof}
\fi
Lemma~\ref{lemma:monomial-acylic} states that the value of a monomial $x$ over the program variables of $\Loop$ only depends on the value of monomials $y$ which precede $x$ in the monomial ordering $\preceq$.
This ensures the dependencies among monomials over the program variables of $\Loop$ to be acyclic.

\subsection{Computing Asymptotic Bounds for Prob-solvable Loops}
\label{subsec:bounds}

The structural result on monomial dependencies from Lemma~\ref{lemma:monomial-acylic} allows for recursive procedures over monomials.
This is exploited in Algorithm~\ref{alg:monom-bound} for computing asymptotic bounds for monomials. 
The standard Big-O notation does not differentiate between positive and negative functions, as it considers the absolute value of functions.
We, however, need to differentiate between functions like $2^i$ and $-2^i$.
Therefore, we introduce the notions of \emph{Domination} and \emph{Bounding Functions}.

\begin{definition}[Domination]
	Let $F$ be a finite set of functions from $\N$ to $\R$.
	A function $g : \N \to \R$ is \emph{dominating} $F$ if eventually $\alpha \cdot g(i) \geq f(i)$ for all $f \in F$ and some $\alpha \in \R^+$.
	A function $g : \N \to \R$ is \emph{dominated by} $F$ if all $f \in F$ dominate $\{g\}$.
\end{definition}
Intuitively, a function $f$ dominates a function $g$ if $f$ eventually surpasses $g$ modulo a positive constant factor.
\emph{Exponential polynomials} are sums of products of polynomials with exponential functions, i.e. $\sum_{j} p_j(x) \cdot c_j^x$, where $c_j \in \R_0^+$.
All functions arising in Algorithms \ref{alg:monom-bound}-\ifshort \ref{alg:r-ast-rule} \else \ref{alg:r-past-rule} \fi are exponential polynomials. 
For a finite set $F$ of exponential polynomials, a function dominating $F$ and a function dominated by $F$ are easily computable with standard techniques, by analyzing the terms of the functions in the finite set $F$.
With $\dominating(F)$ we denote an algorithm computing an exponential polynomial dominating $F$.
With $\dominated(F)$ we denote an algorithm computing an exponential polynomial dominated by $F$.
We assume the functions returned by the algorithms $\dominating(F)$ and $\dominated(F)$ to be monotone and either non-negative or non-positive.

\begin{example}[Domination]
	The following statements are true:
	$0$ dominates $\{ -i^3 + i^2 + 5 \}$, $i^2$ dominates $\{ 2 i^2 \}$, $i^2 \cdot 2^i$ dominates $\{ i^2 \cdot 2^i + i^9, i^5 + i^3, 2^{-i} \}$, $i$ is dominated by $\{ i^2 - 2i + 1, \frac{1}{2} i - 5 \}$ and $-2^{i}$ is dominated by $\{ 2^i - i^2, -10 \cdot 2^{-i}\}$.
\end{example}

\begin{definition}[Bounding Function for $\Loop$]
	Let $E$ be an arithmetic expression over the program variables of $\Loop$.
	Let $l, u : \N \to \R$ be monotone and non-negative or non-positive.
	
	\begin{enumerate}
		\item
		$l$ is a \emph{lower bounding function} for $E$ if eventually $\P(\alpha \cdot l(i) \leq E_i \mid T^{\neg \Guard} > i) = 1$ for some $\alpha \in \R^+$.
		
		\item 
		$u$ is an \emph{upper bounding function} for $E$ if eventually $\P(E_i \leq \alpha \cdot u(i) \mid T^{\neg \Guard} > i) = 1$ for some $\alpha \in \R^+$.
		
		\item
		An \emph{absolute bounding function} for $E$ is an upper bounding function for $|E|$.
	\end{enumerate}
\end{definition}
A bounding function imposes a bound on an expression $E$ over the program variables holding eventually, almost surely, and modulo a positive constant factor.
Moreover, bounds on $E$ only need to hold as long as the program has not yet terminated.

Given a Prob-solvable loop $\Loop$ and a monomial $x$ over the program variables of $\Loop$, Algorithm~\ref{alg:monom-bound} computes a lower and upper bounding function for $x$.
Because every polynomial expression is a linear combination of monomials, the procedure can be used to compute lower and upper bounding functions for any polynomial expression over $\Loop$'s program variables by substituting every monomial with its lower or upper bounding function depending on the sign of the monomial's coefficient.
Once a lower bounding function $l$ and an upper bounding function $u$ are computed, an absolute bounding function can be computed by $\dominating(\{ u, -l \})$.

In Algorithm~\ref{alg:monom-bound}, candidates for bounding functions are modeled using recurrence relations.
Solutions $s(i)$ of these recurrences are closed-form candidates for bounding functions parameterized by loop iteration $i$.
Algorithm~\ref{alg:monom-bound} relies on the existence of closed-form solutions of recurrences.
While closed-forms of general recurrences do not always exist, a property of \emph{C-finite recurrences}, linear recurrences with constant coefficients, is that their closed-forms always exist and are computable~\cite{kauers_concrete_2011}.
In all occurring recurrences, we consider a monomial over program variables as a single function.
Therefore, throughout this section, all recurrences arising from a Prob-solvable loop $\Loop$ in Algorithm~\ref{alg:monom-bound} are C-finite or can be turned into C-finite recurrences.
Moreover, closed-forms $s(i)$ of C-finite recurrences are given by exponential polynomials.
Therefore, for any solution $s(i)$ to a C-finite recurrence and any constant $r \in \R$, the following holds:
\begin{equation}
\exists \alpha, \beta \in \R^+, \exists i_0 \in \N: \forall i \geq i_0 : \alpha \cdot s(i) \leq s(i + r) \leq \beta \cdot s(i).
\end{equation}
Intuitively, the property states that constant shifts do not change the asymptotic behavior of $s$.
We use this property at various proof steps in this section.
Moreover, we recall that limits of exponential polynomials are computable~\cite{gruntz1996}.

For every monomial $x$, every branch $B \in \supp(\Update_\Loop^x)$ is a polynomial over the program variables.
Let $\rec(x) := \{ \text{coefficient of $x$ in $B$} \mid B \in \supp(\Update_\Loop^x) \}$ denote the set of coefficients of the monomial $x$ in all branches of $\Loop$.
Let $\inhom(x) := \{ B - c \cdot x \mid B \in \supp(\Update_\Loop^x) \text{ and $c = $ coefficient of $x$ in $B$}\}$ denote all the branches of the monomial $x$ without $x$ and its coefficient.
The symbolic constants $c_1$ and $c_2$ in Algorithm~\ref{alg:monom-bound} represent arbitrary initial values of the monomial $x$ for which bounding functions are computed.
The fact that they are symbolic ensures that all potential initial values are accounted for.
$c_1$ represents positive initial values and $- c_2$ negative initial values.
The symbolic constant $d$ is used in the recurrences to account for the fact that the bounding functions only hold modulo a constant.
Intuitively, if we use the bounding function in a recurrence we need to restore the lost constant.
$\oversign(x)$ is an over-approximation of the sign of the monomial $x$, i.e., if $\exists i : \P(x_i > 0) > 0$, then $+ \in \oversign(x)$ and if $\exists i : \P(x_i < 0) > 0$, then $- \in \oversign(x)$.

\begin{algorithm}
	\KwIn{A Prob-solvable loop $\Loop$ and a monomial $x$ over $\Loop$'s variables}
	\KwOut{Lower and upper bounding functions $l(i)$, $u(i)$ for $x$}
	
	$\inhomboundsupper := \{ \text{upper bounding function of $P$} \mid P \in \inhom(x) \}$ \hfill {\scriptsize (recursive call)} \label{alg:monom-bound:upper}
	
	$\inhomboundslower := \{ \text{lower bounding function of $P$} \mid P \in \inhom(x) \}$ \hfill {\scriptsize (recursive call)} \label{alg:monom-bound:lower}
	
	$U(i) := \dominating(\inhomboundsupper)$ \label{alg:monom-bound:U}
	
	$L(i) := \dominated(\inhomboundslower)$
	
	$\maxrec := \max \rec(x)$
	
	$\minrec := \min \rec(x)$ \label{alg:monom-bound:minRec}
	
	$I := \emptyset$
	
	\lIf{$+ \in \oversign(x)$}{
		$I := I \union \{c_1\}$
	}
	\lIf{$- \in \oversign(x)$}{
		$I := I \union \{-c_2\}$
	}
	
	$\uppercandidates := $ closed-forms of $\{ y_{i+1} = r \cdot y_{i} + d \cdot U(i) \mid r \in \{ \minrec, \maxrec\}, y_0 \in I\}$ \label{alg:monom-bound:recurrences}
	
	$\lowercandidates := $ closed-forms of $\{ y_{i+1} = r \cdot y_{i} + d \cdot L(i) \mid r \in \{ \minrec, \maxrec\}, y_0 \in I\}$
	
	$u(i) := \dominating(\uppercandidates)$ \label{alg:monom-bound:dominating}
	
	$l(i) := \dominated(\lowercandidates)$ \label{alg:monom-bound:dominated}
	
	\Return $l(i), u(i)$
	
	\caption{Computing bounding functions for monomials}
	\label{alg:monom-bound}
\end{algorithm}

Lemma~\ref{lemma:monomial-acylic}, the computability of closed-forms of C-finite recurrences and the fact that within a Prob-solvable loop only finitely many monomials can occur, implies the termination of Algorithm~\ref{alg:monom-bound}.
Its correctness is stated in the next theorem.

\begin{theorem}[Correctness of Algorithm~\ref{alg:monom-bound}]
	The functions $l(i), u(i)$ returned by Algorithm~\ref{alg:monom-bound} on input $\Loop$ and $x$ are a lower- and an upper bounding function for $x$, respectively.
\end{theorem}
\begin{proof}
	Intuitively, it has to be shown that regardless of the paths through the loop body taken by any program run, the value of $x$ is always eventually upper bounded by some function in $\uppercandidates$ and eventually lower bounded by some function in $\lowercandidates$ (almost surely and modulo positive constant factors).
	We show that $x$ is always eventually upper bounded by some function in $\uppercandidates$.
	The proof for the lower bounding function is analogous.
	
	Let $\vartheta \in \Sigma$ be a \emph{possible} program run, i.e. $\P(Cyl(\pi)) > 0$ for all finite prefixes $\pi$ of $\vartheta$.
	Then, for every $i \in \N$, if $T^{\neg \Guard}(\vartheta) > i$, the following holds:
	\begin{gather*}
		x_{i+1}(\vartheta) = a_{(1)} \cdot x_i(\vartheta) + P_{(1)i}(\vartheta)
		\text{\ \ or\ \ }
		x_{i+1}(\vartheta) = a_{(2)} \cdot x_i(\vartheta) + P_{(2)i}(\vartheta) \\
		\text{ or } ... \text{ or\ \ }
		x_{i+1}(\vartheta) = a_{(k)} \cdot x_i(\vartheta) + P_{(k)i}(\vartheta),
	\end{gather*}
	where $a_{(j)} \in \rec(x)$ and $P_{(j)} \in \inhom(x)$ are polynomials over program variables.
	Let $u_1(i), ..., u_k(i)$ be upper bounding functions of $P_{(1)}, ..., P_{(k)}$, which are computed recursively at line~\ref{alg:monom-bound:recurrences}.
	Moreover, let $U(i) := \dominating(\{u_1(i), ..., u_k(i)\})$, $\minrec = \min \rec(x)$ and $\maxrec = \max \rec(x)$.
	Let $l_0 \in \N$ be the smallest number such that for all $j \in \{1, ..., k\}$ and $i \geq l_0$:
	\begin{gather}
		\label{ps:bounds:b1}
		\P(P_{(j)i} \leq \alpha_j \cdot u_j(i) \mid T^{\neg \Guard} > i) = 1 \text{ for some } \alpha_j \in \R^+, \text{ and } \\
		\label{ps:bounds:b2}
		u_j(i) \leq \beta \cdot U(i) \text{ for some } \beta \in \R^+
	\end{gather}
	Thus, all inequalities from the bounding functions $u_j$ and the dominating function $U$ hold from $l_0$ onward.
	Because $U$ is a dominating function, it is by definition either non-negative or non-positive.
	Assume $U(i)$ to be non-negative, the case for which $U(i)$ is non-positive is symmetric.
	Using the facts \eqref{ps:bounds:b1} and \eqref{ps:bounds:b2}, we establish:
	For the constant $\gamma := \beta \cdot \max_{j=1..k} \alpha_j$, it holds that $\P(P_{(j)i} \leq \gamma \cdot U(i) \mid T^{\neg \Guard} > i) = 1$ for all $j \in \{1, ...,k\}$ and all $i \geq l_0$.
	Let $l_1$ be the smallest number such that $l_1 \geq l_0$ and $U(i + l_0) \leq \delta \cdot U(i)$ for all $i \geq l_1$ and some $\delta \in \R^+$.

	\paragraph{Case 1, $x_i$ is almost surely negative for all $i \geq l_1$:}
	
	Consider the recurrence relation $y_0 = m$, $y_{i+1} = \minrec \cdot y_i + \eta \cdot U(i)$, where $\eta := \max(\gamma, \delta)$ and $m$ is the maximum value of $x_{l_1}(\vartheta)$ among all possible program runs $\vartheta$.
	Note that $m$ exists because there are only finitely many values $x_{l_1}(\vartheta)$ for possible program runs $\vartheta$.
	Moreover, $m$ is negative by our case assumption.
	By induction, we get $\P(x_i \leq y_{i - l_1} \mid T^{\neg \Guard} > i) = 1$ for all $i \geq l_1$.
	Therefore, for a closed-form solution $s(i)$ of the recurrence relation $y_i$, we get $\P(x_i \leq s(i - l_1) \mid T^{\neg \Guard} > i) = 1$ for all $i \geq l_1$.
	We emphasize that $s$ exists and can effectively be computed because $y_i$ is C-finite.
	Moreover, $s(i - l_1) \leq \theta \cdot s(i)$ for all $i \geq l_2$ for some $l_2 \geq l_1$ and some $\theta \in \R^+$.
	Therefore, $s$ satisfies the bound condition of an upper bounding function.
	Also, $s$ is present in $\uppercandidates$ by choosing the symbolic constants $c_2$ and $d$ to represent $-m$ and $\eta$ respectively.
	The function $u(i) := \dominating(\uppercandidates)$, at line~\ref{alg:monom-bound:dominating}, is dominating $\uppercandidates$ (hence also $s$), is monotone and either non-positive or non-negative.
	Therefore, $u(i)$ is an upper bounding function for $x$.

	\paragraph{Case 2, $x_i$ is not almost surely negative for all $i \geq l_1$:}
	
	Thus, there is a possible program run $\vartheta'$ such that $x_i(\vartheta') \geq 0$ for some $i \geq l_1$.
	Let $l_2 \geq l_1$ be the smallest number such that $x_{l_2}(\hat{\vartheta}) \geq 0$ for some possible program run $\hat{\vartheta}$.
	This number certainly exists, as $x_i(\vartheta')$ is non-negative for some $i \geq l_1$.
	Consider the recurrence relation $y_0 = m$, $y_{i+1} = \maxrec \cdot y_i + \eta \cdot U(i)$, where $\eta := \max(\gamma, \delta)$ and $m$ is the maximum value of $x_{l_2}(\vartheta)$ among all possible program runs $\vartheta$.
	Note that $m$ exists because there are only finitely many values $x_{l_2}(\vartheta)$ for possible program runs $\vartheta$.
	Moreover, $m$ is non-negative because $m \geq x_{l_2}(\hat{\vartheta}) \geq 0$.
	By induction, we get $\P(x_i \leq y_{i - l_2} \mid T^{\neg \Guard} > i) = 1$ for all $i \geq l_2$.
	Therefore, for a solution $s(i)$ of the recurrence relation $y_i$, we get $\P(x_i \leq s(i - l_2) \mid T^{\neg \Guard} > i) = 1$ for all $i \geq l_2$.
	As above, $s$ exists and can effectively be computed because $y_i$ is C-finite.
	Moreover, $s(i - l_2) \leq \theta \cdot s(i)$ for all $i \geq l_3$ for some $l_3 \geq l_2$ and some $\theta \in \R^+$.
	Therefore, $s$ satisfies the bound condition of an upper bounding function
	Also, $s$ is present in $\uppercandidates$ by choosing the symbolic constants $c_1$ and $d$ to represent $m$ and $\eta$ respectively.
	The function $u(i) := \dominating(\uppercandidates)$, at line~\ref{alg:monom-bound:dominating}, is dominating $\uppercandidates$ (hence also $s$), is monotone and either non-positive or non-negative.
	Therefore, $u(i)$ is an upper bounding function for $x$.\qed
\end{proof}

\begin{example}[Bounding functions]
	\label{ex:monom-bound}
	We illustrate Algorithm~\ref{alg:monom-bound} by computing bounding functions for $x$ and the Prob-solvable loop from Example~\ref{ex:bounds:rsm}:
	We have $\rec(x) := \{2, \frac{1}{2}\}$ and $\inhom(x) = \{ y^2, 0 \}$.
	Computing bounding functions recursively for $P \in \inhom(x) = \{ y^2, 0 \}$ is simple, as we can give exact bounds leading to $\inhomboundsupper = \{i^2, 0\}$ and $\inhomboundslower = \{i^2, 0\}$.
	Consequently, we get $U(i) = i^2$, $L(i) = 0$, $\maxrec = 2$ and $\minrec = \frac{1}{2}$.
	With a rudimentary static analysis of the loop, we determine the (exact) over-approximation $Sign(x) := \{+\}$ by observing that $x_0 > 0$ and all $P \in \inhom(x)$ are strictly positive.
	Therefore, $\uppercandidates$ is the set of closed-form solutions of the recurrences $y_0 := c_1$, $y_{i+1} := 2 y_i + d \cdot i^2$ and $y_0 := c_1$, $y_{i+1} := \frac{1}{2} y_i + d \cdot i^2$.
	Similarly, $\lowercandidates$ is the set of closed-form solutions of the recurrences $y_0 := c_1$, $y_{i+1} := 2 y_i$ and $y_0 := c_1$, $y_{i+1} := \frac{1}{2} y_i$.
	Using any algorithm for computing closed-forms of C-finite recurrences, we obtain $\uppercandidates = \{c_1 2^i - d i^2 - 2 d i + 3 d 2^i - 3 d,\ c_1 2^{-i} + 2 d i^2 - 8 d i - 12 d 2^{-i} + 12 d\}$ and $\lowercandidates = \{c_1 2^i,\ c_1 2^{-i} \}$.
	This leads to the upper bounding function $u(i) = 2^i$ and the lower bounding function $l(i) = 2^{-i}$.
	The bounding functions $l(i)$ and $u(i)$ can be used to compute bounding functions for expressions containing $x$ linearly by replacing $x$ by $l(i)$ or $u(i)$ depending on the sign of the coefficient of $x$.
	For instance, eventually and almost surely the following inequality holds:
	$- \frac{x_i}{4} -\frac{i^2}{2} - i - \frac{1}{2} \leq - \frac{1}{4} \cdot \alpha \cdot 2^{-i} -\frac{i^2}{2} - i - \frac{1}{2}$
	for some $\alpha \in \R^+$.
	The inequality results from replacing $x_i$ by $l(i)$.
	Therefore, eventually and almost surely $-\frac{x_i}{4} -\frac{i^2}{2} - i - \frac{1}{2} \leq - \beta \cdot i^2$ for some $\beta \in \R^+$.
	Thus, $-i^2$ is an upper bounding function for the expression $- \frac{x_i}{4} -\frac{i^2}{2} - i - \frac{1}{2}$.
\end{example}

\begin{remark}
	Algorithm~\ref{alg:monom-bound} describes a general procedure computing bounding functions for special sequences.
	Figuratively, that is for sequences $s$ such that $s_{i+1} = f(s_i, i)$ but in every step the function $f$ is chosen non-deterministically among a fixed set of special functions (corresponding to branches in our case).
	We reserve the investigation of applications of bounding functions for such sequences beyond the probabilistic setting for future work.
\end{remark}

\subsection{Algorithms for Termination Analysis of Prob-solvable Loops}
\label{subsec:algorithms-proof-rules}

Using Algorithm~\ref{alg:monom-bound} to compute bounding functions for polynomial expressions over program variables at hand, we are now able to formalize our algorithmic approaches automating the termination analysis of Prob-solvable loops using the proof rules from Section~\ref{section:relaxations}.
Given a Prob-solvable loop $\Loop$ and a polynomial expression $E$ over $\Loop$'s variables, we denote with $\lbf(E)$, $\ubf(E)$ and $\abf(E)$ functions computing a lower, upper and absolute bounding function for $E$ respectively.
Our algorithmic approach for proving \gls{past} using the \gls{rsm-rule} is given in Algorithm~\ref{alg:rsm-rule}.
\begin{algorithm}
	\KwIn{Prob-solvable loop $\Loop$}
	\KwOut{If \emph{true} then $\Loop$ with $G$ satisfies the \gls{rsm-rule}; hence $\Loop$ is \gls{past}}
	
	$E := \E(G_{i+1} - G_i \mid \F_i)$ \label{alg:rsm-rule:E}
	
	$u(i) := \ubf(E)$ \label{alg:rsm-rule:ubf}
	
	$\limit := \lim_{i \to \infty} u(i)$
	
	\Return $\limit < 0$ \label{alg:rsm-rule:true}
	
	\caption{Ranking-Supermartingale-Rule for proving \gls{past}}
	\label{alg:rsm-rule}
\end{algorithm}

\begin{example}[Algorithm~\ref{alg:rsm-rule}]
	Let us illustrate Algorithm~\ref{alg:rsm-rule} with the Prob-solvable loop from Examples \ref{ex:bounds:rsm} and \ref{ex:monom-bound}.
	Applying Algorithm~\ref{alg:rsm-rule} on $\Loop$ leads to $E = - \frac{x_i}{4} -\frac{i^2}{2} - i - \frac{1}{2}$.
	We obtain the upper bounding function $u(i) := -i^2$ for $E$.
	Because $\lim_{i \to \infty} u(i) < 0$, Algorithm~\ref{alg:rsm-rule} returns true.
	This is valid because $u(i)$ having a negative limit witnesses that $E$ is eventually bounded by a negative constant and therefore is eventually an \gls{rsm}.
\end{example}

We recall that all functions arising from $\Loop$ are exponential polynomials (see Section~\ref{subsec:bounds}) and that limits of exponential polynomials are computable~\cite{gruntz1996}. 
Therefore, the termination of Algorithm~\ref{alg:rsm-rule} is guaranteed and its correctness is stated next.

\begin{theorem}[Correctness of Algorithm~\ref{alg:rsm-rule}] 
	\label{th:correctness:rsm-rule}
	If Algorithm~\ref{alg:rsm-rule} returns \emph{true} on input $\Loop$, then $\Loop$ with $G_\Loop$ satisfies the \gls{rsm-rule}.
\end{theorem}
\begin{proof}
	When returning \emph{true} at line~\ref{alg:rsm-rule:true} we have $\P(E_i \leq \alpha \cdot u(i) \mid T^{\neg \Guard} > i) = 1$ for all $i \geq i_0$ and some $i_0 \in \N$, $\alpha \in \R^+$.
	Moreover, $u(i) < -\epsilon$ for all $i \geq i_1$ for some $i_1 \in \N$, by the definition of $\lim$.
	From this follows that $\forall i \geq \max(i_0, i_1)$ almost surely $\Guard_i \implies \E(G_{i+1} - G_i \mid \F_i) \leq - \alpha {\cdot} \epsilon$, which means $G$ is eventually an \gls{rsm}.\qed
\end{proof}

Our approach proving \gls{ast} using the \gls{sm-rule} is captured with Algorithm~\ref{alg:sm-rule}.

\begin{algorithm}
	\KwIn{Prob-solvable loop $\Loop$}
	\KwOut{If \emph{true}, $\Loop$ with $G$ satisfies the \gls{sm-rule} with constant $d$ and $p$; hence $\Loop$ is \gls{ast}}
	
	$E := \E(G_{i+1} - G_i \mid \F_i)$ \label{alg:sm-rule:E}
	
	$u(i) := \ubf(E)$ \label{alg:sm-rule:u}
	
	\lIf{not eventually $u(i) \leq 0$}{ \label{alg:sm-rule:if}
		\Return false
	}
	
	\For{$B \in \supp(\Update_\Loop^G)$}{ \label{alg:sm-rule:for}
		
		$d(i) := \ubf(B - G)$ \label{alg:sm-rule:d}
		
		$\limit := \lim_{i \to \infty} d(i)$
		
		\lIf{$\limit < 0$}{
			\Return true \label{alg:sm-rule:true}
		}
	}
	
	\Return false
	
	\caption{Supermartingale-Rule for proving \gls{ast}}
	\label{alg:sm-rule}
\end{algorithm}

\begin{example}[Algorithm~\ref{alg:sm-rule}]
	Let us illustrate Algorithm~\ref{alg:sm-rule} for the Prob-solvable loop $\Loop$ from Figure~\ref{fig:limit-sm-rule}:
	Applying Algorithm~\ref{alg:sm-rule} on $\Loop$ yields $E \equiv 0$ and $u(i) = 0$.
	The expression $G$ ($=x$) has two branches.
	One of them is $x_i - y_i + 4$, which occurs with probability $\nicefrac{1}{2}$.
	When the for-loop of Algorithm~\ref{alg:sm-rule} reaches this branch $B = x_i - y_i + 4$ on line~\ref{alg:sm-rule:for}, it computes the difference $B - G = -y_i + 4$.
	An upper bounding function for $B - G$ is given by $d(i) = -i$.
	Because $\lim_{i \to \infty} d(i) < 0$, Algorithm~\ref{alg:sm-rule} returns true.
	This is valid because of the branch $B$ witnessing that $G$ eventually decreases by at least a constant with probability $\nicefrac{1}{2}$.
	Therefore, all conditions of the \gls{sm-rule} are satisfied and $\Loop$ is \gls{ast}.
\end{example}

\begin{theorem}[Correctness of Algorithm~\ref{alg:sm-rule}]
	\label{th:correctness:sm-rule}
	If Algorithm~\ref{alg:sm-rule} returns \emph{true} on input $\Loop$, then $\Loop$ with $G_\Loop$ satisfies the \gls{sm-rule} with constant $d$ and $p$.
\end{theorem}
\ifshort
The proof of Theorem~\ref{th:correctness:sm-rule}, as well as of
Theorem~\ref{th:correctness:r-ast-rule},
are similar to the one of Theorem~\ref{th:correctness:rsm-rule} and 
can be found in~\cite{moosbrugger2020automated}.
\else
\begin{proof}
	Similarly as for the correctness of Algorithm~\ref{alg:rsm-rule}, $G$ is a supermartingale if Algorithm~\ref{alg:sm-rule} returns \emph{true}.
	Moreover, there is a branch $B \in \supp(\Update_\Loop^G)$ such that $G$ changes eventually and almost surely by at most $\alpha \cdot d(i)$, for some $\alpha \in \R^+$.
	In addition, because $\lim_{i \to \infty} d(i) < 0$, it follows that $d(i) \leq - \epsilon$ for all $i \geq i_0$ for some $i_0 \in \N$, $\epsilon \in \R^+$.
	Therefore, eventually $G$ decreases by at least $\alpha \cdot \epsilon$ with probability at least $\Update_\Loop^G(B) > 0$.
	Hence, all conditions of the \gls{sm-rule} are satisfied.\qed
\end{proof}
\fi

As established in Section~\ref{section:relaxations}, the relaxation of the \gls{r-ast-rule} requires that there is a positive probability of reaching the iteration $i_0$ after which the conditions of the proof rule hold.
Regarding automation, we strengthen this condition by ensuring that there is a positive probability of reaching any iteration, i.e. $\forall i \in \N : \P(\Guard_i) > 0$.
Obviously, this implies $\P(\Guard_{i_0}) > 0$.
Furthermore, with $\canreachanyiter(\Loop)$ we denote a computable under-approximation of $\forall i \in \N : \P(\Guard_i) > 0$.
That means, $\canreachanyiter(\Loop)$ implies $\forall i \in \N : \P(\Guard_i) > 0$.
Our approach proving non-\gls{ast} is summarized in Algorithm~\ref{alg:r-ast-rule}.

\begin{algorithm}
	\KwIn{Prob-solvable loop $\Loop$}
	\KwOut{if \emph{true}, $\Loop$ with $-G$ satisfies the \gls{r-ast-rule}; hence $\Loop$ is not \gls{ast}}
	
	$E := \E(-G_{i+1} + G_i \mid \F_i)$ \label{alg:r-ast-rule:E}
	
	$u(i) := \ubf(E)$ \label{alg:r-ast-rule:u}
	
	\lIf{not eventually $u(i) \leq 0$}{ \label{alg:r-ast-rule:if-u}
		\Return false
	}
	
	\lIf{$\neg \canreachanyiter(\Loop)$}{ \label{alg:r-ast-rule:if-c}
		\Return false
	}
	
	$\epsilon(i) := - u(i)$ \label{alg:r-ast-rule:e}
	
	\lIf{$\epsilon(i) \not\in \Omega(1)$}{ \label{alg:r-ast-rule:if-e-const}
		\Return false
	}
	
	$\differences := \{ B + G \mid B \in \supp(\Update_\Loop^{-G}) \}$ \label{alg:r-ast-rule:differences}
	
	$\diffbounds := \{ \abf(d) \mid d \in \differences \}$ \label{alg:r-ast-rule:bounds}
	
	$c(i) := \dominating(\diffbounds)$ \label{alg:r-ast-rule:c_i}
	
	\Return $c(i) \in O(1)$ \label{alg:r-ast-rule:true}
	
	\caption{Repulsing-AST-Rule for proving non-\gls{ast}}
	\label{alg:r-ast-rule}
\end{algorithm}

\begin{example}[Algorithm~\ref{alg:r-ast-rule}]
	Let us illustrate Algorithm~\ref{alg:r-ast-rule} for the Prob-solvable loop $\Loop$ from Figure~\ref{fig:limit-sm-rule}:
	Applying Algorithm~\ref{alg:r-ast-rule} on $\Loop$ leads to $E = \frac{y_i}{6} - \frac{1}{3} = \frac{2^{-i}}{3} - \frac{1}{3}$ and to the upper bounding function $u(i) = -1$ for $E$ on line~\ref{alg:r-ast-rule:u}.
	Therefore, the if-statement on line~\ref{alg:r-ast-rule:if-u} is not executed, which means $-G$ is eventually a $\epsilon$-repulsing supermartingale.
	Moreover, with a simple static analysis of the loop, we establish $\canreachanyiter(\Loop)$ to be true, as there is a positive probability that the loop guard does not decrease.
	Thus, the if-statement on line~\ref{alg:r-ast-rule:if-c} is not executed.
	Also, the if-statement on line~\ref{alg:r-ast-rule:if-e-const} is not executed, because $\epsilon(i) = -u(i) = 1$ is constant and therefore in $\Omega(1)$.
	$E$ eventually decreases by $\epsilon = 1$ (modulo a positive constant factor), because $u(i) = -1$ is an upper bounding function for $E$.
	We have $\differences = \{ 1 - \frac{y_i}{2}, 1 + \frac{y_i}{2} \}$.
	Both expressions in $\differences$ have an absolute bounding function of $1$.
	Therefore, $\diffbounds = \{1\}$.
	As a result on line~\ref{alg:r-ast-rule:c_i} we have $c(i) = 1$, which eventually and almost surely is an upper bound on $|-G_{i+1} + G_i|$ (modulo a positive constant factor).
	Therefore, the algorithm returns true.
	This is correct, as all the preconditions of the \gls{r-ast-rule} are satisfied (and therefore $\Loop$ is not \gls{ast}).
\end{example}

\begin{theorem}[Correctness of Algorithm~\ref{alg:r-ast-rule}]
	\label{th:correctness:r-ast-rule}
	If Algorithm~\ref{alg:r-ast-rule} returns \emph{true} on input $\Loop$, then $\Loop$ with $-G_\Loop$ satisfies the \gls{r-ast-rule}.
\end{theorem}
\ifshort\else
\begin{proof}
	With the same reasoning as for the correctness of Algorithm~\ref{alg:sm-rule}, $-G$ is a supermartingale if Algorithm~\ref{alg:r-ast-rule} returns \emph{true}.
	Moreover, the condition $\P(-G_{i_0} < 0) > 0$ of the \gls{r-ast-rule} is satisfied, due to the under-approximation \\$\canreachanyiter(\Loop)$ and the if-statement on line \ref{alg:r-ast-rule:if-c}.
	The function $u(i)$ is an upper bounding function for $\E(-G_{i+1} + G_i \mid \F_i)$.
	Hence, eventually and almost surely $\E(-G_{i+1} + G_i \mid \F_i) \leq - \alpha \cdot \epsilon(i)$ for $\epsilon(i) := -u(i)$ and some $\alpha \in \R^+$.
	The if-statement at line~\ref{alg:r-ast-rule:if-e-const} ensures that $\epsilon(i)$ is lower bounded by a constant.
	Therefore, $-G$ eventually is an $(\alpha \cdot \epsilon)$-repulsing supermartingale.
	The function $c(i)$, assigned to $\dominating(\diffbounds)$, is a function dominating absolute bounding functions of all branches of $-G_{i+1} + G_i$.
	Consequently, $c(i)$ is a bound on the differences of $G$, i.e. eventually and almost surely $|-G_{i+1} + G_i| \leq \beta \cdot c(i)$ for some $\beta \in \R^+$.
	Algorithm~\ref{alg:r-ast-rule} returns true only if $c(i)$ can be bounded by a constant which in turn means $G$ has $(\beta \cdot c)$-bounded differences.
	Thus, if Algorithm~\ref{alg:r-ast-rule} returns \emph{true}, all preconditions of the \gls{r-ast-rule} are satisfied.\qed
\end{proof}
\fi

\ifshort
Because the \gls{r-past-rule} is a slight variation of the \gls{r-ast-rule}, Algorithm~\ref{alg:r-ast-rule} can be slightly modified to yield a procedure for the \gls{r-past-rule}.
An algorithm for the \gls{r-past-rule} is provided in \cite{moosbrugger2020automated}.
\fi

\ifshort\else
We finally provide Algorithm~\ref{alg:r-past-rule} for the \gls{r-past-rule}.
The algorithm is a variation of Algorithm~\ref{alg:r-ast-rule} (for the \gls{r-ast-rule}).
The if-statement on line~\ref{alg:r-past-rule:mart} forces $-G$ to be a martingale.
Therefore, after the if-statement $-G$ is an $\epsilon$-repulsing supermartingale with $\epsilon = 0$.

\begin{algorithm}
	\KwIn{Prob-solvable loop $\Loop$}
	\KwOut{If \emph{true}, $\Loop$ with $-G$ satisfies the \gls{r-past-rule}; hence $\Loop$ is not \gls{past}}
	
	$E := \E(-G_{i+1} + G_i \mid \F_i)$
	
	\lIf{$E \not\equiv 0$}{ \label{alg:r-past-rule:mart}
		\Return false
	}
	
	\lIf{$\neg \canreachanyiter(\Loop)$}{
		\Return false
	}
	
	$\differences := \{ B + G \mid B \in \supp(\Update_\Loop^{-G}) \}$
	
	$\diffbounds := \{ \abf(d) \mid d \in \differences \}$
	
	$c(i) := \dominating(\diffbounds)$
	
	\Return $c(i) \in O(1)$ \label{alg:r-past-rule:const}
	
	\caption{Repulsing-PAST-Rule for proving non-\gls{past}}
	\label{alg:r-past-rule}
\end{algorithm}
\fi

\subsection{Ruling out Proof Rules for Prob-Solvable Loops}
\label{subsec:rule-out-rules}

A question arising when combining our algorithmic approaches from Section~\ref{subsec:algorithms-proof-rules} into a unifying framework is that, given a Prob-solvable loop $\Loop$, what algorithm to apply first for determining $\Loop$'s termination behavior?
In~\cite{bartocci_automatic_2019} the authors provide an algorithm for computing an algebraically closed-form of $\E(M_i)$, where $M$ is a polynomial over $\Loop$'s variables.
The following lemma explains how the expression $\E(M_{i+1} - M_i)$ relates to the expression $\E(M_{i+1} - M_i \mid \F_i)$.
\ifshort
The lemma follows from the monotonicity of $\E$.
\fi

\begin{lemma}[Rule out Rules for $\Loop$]
	\label{lemma:rule-out-rules}
	Let $(M_i)_{i \in \N}$ be a stochastic process.
	If $\E(M_{i+1} - M_i \mid \F_i) \leq - \epsilon$ then $\E(M_{i+1} - M_i) \leq - \epsilon$, for any $\epsilon \in \R^+$.
\end{lemma}
\ifshort\else
\begin{proof}
	\begin{flalign*}
		&\E(M_{i+1} - M_i \mid \F_i) \leq - \epsilon &\implies && \hfill \text{(Monotonicity of $\E$)} \\
		&\E(\E(M_{i+1} - M_i \mid \F_i)) \leq \E(-\epsilon) &\iff && \hfill \text{(Property of $\E( \cdot \mid \F_i)$)} \\
		&\E(M_{i+1} - M_i) \leq \E(-\epsilon) &\iff && \hfill \text{($-\epsilon$ is constant)} \\
		&\E(M_{i+1} - M_i) \leq -\epsilon & && \hfill \qed
	\end{flalign*}
\end{proof}
\fi
The contrapositive of Lemma~\ref{lemma:rule-out-rules} provides a criterion to rule out the viability of a given proof rule.
For a Prob-solvable loop $\Loop$, if $\E(G_{i+1} - G_i) \not\leq 0$ then $\E(G_{i+1} - G_i \mid \F_i) \not\leq 0$, meaning $G$ is not a supermartingale.
The expression $\E(G_{i+1} - G_i)$ depends only on $i$ and can be computed by $\E(G_{i+1} - G_i) = \E(G_{i+1}) - \E(G_i)$, where the expected value $\E(G_i)$ is computed as in~\cite{bartocci_automatic_2019}.
Therefore, in some cases, proof rules can automatically be deemed nonviable, without the need to compute bounding functions.
	
	\section{Implementation and Evaluation}
\label{section:implementation}

\subsection{Implementation}

We implemented and combined our algorithmic approaches from Section~\ref{section:algorithms} in the new software tool \amber{} to stand for \emph{Asymptotic Martingale Bounds}.
\amber{} and all benchmarks are available at \url{https://github.com/probing-lab/amber}.
\amber{} uses \textsc{Mora}~\cite{bartocci_automatic_2019}\cite{BartocciKS20} for computing the first-order moments of program variables and the \textsc{diofant} package\footnote{\url{https://github.com/diofant/diofant}} as its computer algebra system.

\paragraph{Computing $\dominating$ and $\dominated$}
The $\dominating$ and $\dominated$ procedures used in Algorithms~\ref{alg:monom-bound} and \ref{alg:r-ast-rule} are implemented by combining standard algorithms for Big-O analysis and bookkeeping of the asymptotic polarity of the input functions.
Let us illustrate this.
Consider the following two input-output-pairs which our implementation would produce: (a) $\dominating(\{i^2 + 10, 10 \cdot i^5 - i^3\}) = i^5$ and (b) $\dominating(\{-i + 50, -i^8 + i^2 - 3 \cdot i^3\}) = -i$.
For (a) $i^5$ is eventually greater than all functions in the input set modulo a constant factor because all functions in the input set are $O(i^5)$.
Therefore, $i^5$ dominates the input set.
For (b), the first function is $O(i)$ and the second is $O(i^8)$.
In this case, however, both functions are eventually negative.
Therefore, $-i$ is a function dominating the input set.
Important is the fact that an exponential polynomial $\sum_{j} p_j(i) \cdot c_j^i$, where $c_j \in \R_0^+$ will always be eventually either only positive or only negative (or $0$ if identical to $0$).

\paragraph{Sign Over-Approximation}
The over-approximation $\oversign(x)$ of the signs of a monomial $x$ used in Algorithm~\ref{alg:monom-bound} is implemented by a simple static analysis:
For a monomial $x$ consisting solely of even powers, $\oversign(x) = \{+\}$.
For a general monomial $x$, if $x_0 \geq 0$ and all monomials on which $x$ depends, together with their associated coefficients are always positive, then $- \not\in \oversign(x)$.
For example, if $\supp(\Update_\Loop^x) = \{ x_i + 2y_i - 3z_i, x_i + u_i \}$, then $- \not\in \oversign(x)$ if $x_0 \geq 0$ as well as $- \not\in \oversign(y)$, $+ \not\in \oversign(z)$ and $- \not\in \oversign(u)$.
Otherwise, $- \in \oversign(x)$.
The over-approximation for $+ \not\in \oversign(x)$ is analogous.

\paragraph{Reachability Under-Approximation}
$\canreachanyiter(\Loop)$, used in Algorithm~\ref{alg:r-ast-rule}, needs to satisfy the property that if it returns true, then loop $\Loop$ reaches any iteration with positive probability.
In \amber{}, we implement this under-approximation as follows:
$\canreachanyiter(\Loop)$ is true if there is a branch $B$ of the loop guard polynomial $G_\Loop$ such that $B - G_{\Loop i}$ is non-negative for all $i \in \N$.
Otherwise, $\canreachanyiter(\Loop)$ is false.
In other words, if $\canreachanyiter(\Loop)$ is true, then in any iteration there is a positive probability of $G_\Loop$ not decreasing.

\paragraph{Bound Computation Improvements}
In addition to Algorithm~\ref{alg:monom-bound} computing bounding
functions for monomials of program variables, \amber{} implements the following refinements:
\begin{enumerate}
	\item
	A monomial $x$ is deterministic, which means it is independent of probabilistic choices, if $x$ has a single branch and only depends on monomials having single branches.
	In this case, the exact value of $x$ in any iteration is given by its first-order moments and bounding functions can be obtained by using these exact representations.
	
	\item Bounding functions for an odd power $p$ of a monomial $x$ can be computed by $u(i)^p$ and $l(i)^p$, where $u(i)$ is an upper- and $l(i)$ a lower bounding function for $x$.
\end{enumerate}
Whenever the above enhancements are applicable, \amber{} prefers them over Algorithm~\ref{alg:monom-bound}.

\subsection{Experimental Setting and Results}

\paragraph{Experimental Setting and Comparisons}
Regarding {programs which are \gls{past}}, we compare \amber{} against the tool \textsc{Absynth}~\cite{ngo_bounded_2018} and the tool in~\cite{chakarov_probabilistic_2013} which we refer to as \textsc{Mgen}.
\textsc{Absynth} uses a system of inference rules over the syntax of probabilistic programs to derive bounds on the expected resource consumption of a program and can, therefore, be used to certify \gls{past}.
In comparison to \amber{}, \textsc{Absynth} requires the degree of the bound to be provided upfront.
Moreover, \textsc{Absynth} cannot refute the existence of a bound and therefore cannot handle programs that are not \gls{past}.
\textsc{Mgen} uses linear programming to synthesize linear martingales and supermartingales for probabilistic transition systems with linear variable updates.
To certify \gls{past}, we extended \textsc{Mgen}~\cite{chakarov_probabilistic_2013} with the SMT solver \textsc{Z3}~\cite{MouraB08} in order to find or refute
the existence of conical combinations of the (super)martingales derived
by \textsc{Mgen} which yield \glspl{rsm}.

With \amberlight{} we refer to a variant of \amber{} without the relaxations of the proof rules introduced in Section~\ref{section:relaxations}.
That is, with \amberlight{} the conditions of the proof rules need to hold for all $i \in \N$, whereas with \amber{} the conditions are allowed to only hold eventually.
For all benchmarks, we compare \amber{} against \amberlight{} to show the effectiveness of the respective relaxations.
For each experimental table (Tables~\ref{tab:bench:past}-\ref{tab:bench:nast}), \benchsucc{} symbolizes that the respective tool successfully certified \gls{past}/\gls{ast}/non-\gls{ast} for the given program;
\benchfail{} means it failed to certify \gls{past}/\gls{ast}/non-\gls{ast}.
Further, \benchna{} indicates the respective tool failed to certify  \gls{past}/\gls{ast}/non-\gls{ast} because the given program is out-of-scope of the tool's capabilities.
Every benchmark has been run on a machine with a $2.2$ GHz Intel i7 (Gen 6) processor and 16 GB of RAM and finished within a timeout of $50$ seconds, where most benchmarks terminated within a few seconds.

\paragraph{Benchmarks}
We evaluated \amber{} against 38 probabilistic programs.
We present our experimental results by separating our benchmarks within three categories: (i) 21 programs which are \gls{past} (Table~\ref{tab:bench:past}), (ii) 11 programs which are \gls{ast} (Table~\ref{tab:bench:ast}) but not necessarily \gls{past}, and (iii) 6 programs which are not \gls{ast} (Table~\ref{tab:bench:nast}).
The benchmarks have either been introduced in the literature on probabilistic programming \cite{ngo_bounded_2018,chakarov_probabilistic_2013,bartocci_automatic_2019,GieslGH19,mciver_new_2017}, are adaptations of well-known stochastic processes or have been designed specifically to test unique features of \amber{}, like the ability to handle polynomial real arithmetic.

The 21 \gls{past} benchmarks consist of 10 programs representing the original benchmarks of \textsc{Mgen}~\cite{chakarov_probabilistic_2013} and \textsc{Absynth}~\cite{ngo_bounded_2018} augmented with 11 additional probabilistic programs.
Not all benchmarks of \textsc{Mgen} and \textsc{Absynth} could be used for our comparison as \textsc{Mgen} and \textsc{Absynth} target related but different computation tasks than certifying \gls{past}.
Namely, \textsc{Mgen} aims to synthesize (super)martingales, but not ranking ones, whereas \textsc{Absynth} focuses on computing bounds on the expected runtime.
Therefore, we adopted \emph{all} (50) benchmarks from \cite{chakarov_probabilistic_2013} (11) and \cite{ngo_bounded_2018} (39) for which the termination behavior is non-trivial.
A benchmark is trivial regarding \gls{past} if either (i) there is no loop, (ii) the loop is bounded by a constant, or (iii) the program is meant to run forever.
Moreover, we cleansed the benchmarks of programs for which the witness for \gls{past} is just a trivial combination of witnesses for already included programs.
For instance, the benchmarks of \cite{ngo_bounded_2018} contain multiple programs that are concatenated constant biased-random-walks.
These are relevant benchmarks when evaluating \textsc{Absynth} for discovering bounds, but would blur the picture when comparing against \amber{} for \gls{past} certification.
With these criteria, 10 out of the 50 original benchmarks of \cite{chakarov_probabilistic_2013} and \cite{ngo_bounded_2018} remain.
We add 11 additional benchmarks which have either been introduced in the literature on probabilistic programming \cite{bartocci_automatic_2019,GieslGH19,mciver_new_2017}, are adaptations of well-known stochastic processes or have been designed specifically to test unique features of \amber{}.
Notably, out of the 50 original benchmarks from \cite{ngo_bounded_2018} and \cite{chakarov_probabilistic_2013}, only 2 remain which are included in our benchmarks and which \amber{} cannot prove \gls{past} (because they are not Prob-solvable).
All our benchmarks are available at \url{https://github.com/probing-lab/amber}.

\begin{table}[t]
	\tablesize
	\centering
	\bgroup
	\def\arraystretch{1.5}
	\begin{tabular}{lcccc}
		\toprule
		Program & \rotatebox{\tablerotation}{\amber{}} & \rotatebox{\tablerotation}{\amberlight{}} & \rotatebox{\tablerotation}{\textsc{Absynth}} & \rotatebox{\tablerotation}{\textsc{Mgen}+\textsc{Z3}} \\
		\midrule
		2d\_bounded\_random\_walk & \benchsucc{} & \benchsucc{} & \benchfail{} & \benchna{} \\
		\hdashline
		biased\_random\_walk\_constant & \benchsucc{} & \benchsucc{} & \benchsucc{} & \benchsucc{} \\
		\hdashline
		biased\_random\_walk\_exp & \benchsucc{} & \benchsucc{} & \benchfail{} & \benchsucc \\
		\hdashline
		biased\_random\_walk\_poly & \benchsucc{} & \benchfail{} & \benchfail{} & \benchfail \\
		\hdashline
		binomial\_past & \benchsucc{} & \benchsucc{} & \benchsucc{} & \benchsucc{} \\
		\hdashline
		complex\_past & \benchsucc{} & \benchfail{} & \benchfail{} & \benchna{} \\
		\hdashline
		consecutive\_bernoulli\_trails & \benchsucc{} & \benchsucc{} & \benchsucc{} & \benchsucc{} \\
		\hdashline
		coupon\_collector\_4 & \benchsucc{} & \benchfail{} & \benchfail{} & \benchsucc{} \\
		\hdashline
		coupon\_collector\_5 & \benchsucc{} & \benchfail{} & \benchfail{} & \benchsucc{} \\
		\hdashline
		dueling\_cowboys & \benchsucc{} & \benchsucc{} & \benchsucc{} & \benchsucc{} \\
		\hdashline
		exponential\_past\_1 & \benchsucc{} & \benchsucc{} & \benchna{} & \benchna{} \\
		\hdashline
	\end{tabular}\hfill
	\begin{tabular}{lcccc}
		\noalign{\vskip 0.9em}
		\midrule
		Program & \rotatebox{\tablerotation}{\amber{}} & \rotatebox{\tablerotation}{\amberlight{}} & \rotatebox{\tablerotation}{\textsc{Absynth}} & \rotatebox{\tablerotation}{\textsc{Mgen}+\textsc{Z3}} \\
		\midrule
		exponential\_past\_2 & \benchsucc{} & \benchsucc{} & \benchna{} & \benchna{} \\
		\hdashline
		geometric & \benchsucc{} & \benchsucc{} & \benchsucc{} & \benchsucc{} \\
		\hdashline
		geometric\_exponential & \benchfail{} & \benchfail{} & \benchfail{} & \benchfail{} \\
		\hdashline
		linear\_past\_1 & \benchsucc{} & \benchsucc{} & \benchfail{} & \benchfail{} \\
		\hdashline
		linear\_past\_2 & \benchsucc{} & \benchsucc{} & \benchfail{} & \benchna{} \\
		\hdashline
		nested\_loops & \benchna{} & \benchna{} & \benchsucc & \benchfail{} \\
		\hdashline
		polynomial\_past\_1 & \benchsucc{} & \benchfail{} & \benchfail{} & \benchna{} \\
		\hdashline
		polynomial\_past\_2 & \benchsucc{} & \benchfail{} & \benchfail{} & \benchna{} \\
		\hdashline
		sequential\_loops & \benchna{} & \benchna{} & \benchsucc{} & \benchfail{} \\
		\hdashline
		tortoise\_hare\_race & \benchsucc{} & \benchsucc{} & \benchsucc{} & \benchsucc \\
		\midrule
		Total \benchsucc{} & 18 & 12 & 8 & 9 \\
		\bottomrule
	\end{tabular}
	\egroup
	\vspace{1em}
	\caption{21 programs which are \gls{past}.}
	\label{tab:bench:past}
\end{table}

\paragraph{Experiments with \gls{past} -- Table~\ref{tab:bench:past}:}
Out of the 21 \gls{past} benchmarks, \amber{} certifies 18 programs.
\amber{} cannot handle the benchmarks \emph{nested\_loops} and \emph{sequential\_loops}, as these examples use nested or sequential loops and thus are not expressible as Prob-solvable loops.
The benchmarks \emph{exponential\_past\_1} and \emph{exponential\_past\_2} are out of scope of \textsc{Absynth} because they require real numbers, while \textsc{Absynth} can only handle integers.
\textsc{Mgen}+\textsc{Z3} cannot handle benchmarks containing non-linear variable updates or non-linear guards.
Table~\ref{tab:bench:past} shows that \amber{} outperforms both \textsc{Absynth} and \textsc{Mgen}+\textsc{Z3} for Prob-solvable loops, even when our relaxed proof rules from Section~\ref{section:relaxations} are not used.
Yet, our experiments show that our relaxed proof rules enable \amber{} to certify 6
examples to be \gls{past}, which could not be proved without these relaxations by \amberlight{}.

\paragraph{Experiments with \gls{ast} -- Table~\ref{tab:bench:ast}:}
We compare \amber{} against \amberlight{} on 11 benchmarks which are \gls{ast} but not necessarily \gls{past} and also cannot be split into \gls{past} subprograms.
Therefore, the \gls{sm-rule} is needed to certify \gls{ast}.
To the best of our knowledge, \amber{} is the first tool able to certify \gls{ast} for such programs.
Existing approaches like \cite{agrawal_lexicographic_2017} and \cite{ChenH20} can only witness \gls{ast} for non-\gls{past} programs, if - intuitively speaking - the programs contain subprograms which are \gls{past}.
Therefore, we compared \amber{} only against \amberlight{} on this set of examples.
The benchmark \emph{symmetric\_2d\_random\_walk}, which \amber{} fails to certify as \gls{ast}, models the symmetric random walk in $\R^2$ and is still out of reach of current automation techniques.
In \cite{mciver_new_2017} the authors mention that a closed-form expression $M$ and functions $p$ and $d$ satisfying the conditions of the \gls{sm-rule} have not been discovered yet.
The benchmark \emph{fair\_in\_limit\_random\_walk} involves non-constant probabilities and can therefore not be modeled as a Prob-solvable loop.

\paragraph{Experiments with non-\gls{ast} -- Table~\ref{tab:bench:nast}:} 
We compare \amber{} against \amberlight{} on 6 benchmarks which are not \gls{ast}.
To the best of our knowledge, \amber{} is the first tool able to certify non-\gls{ast} for such programs, and thus we compared \amber{} only against \amberlight{}. 
In~\cite{chatterjee_stochastic_2017}, where the notion of repulsing supermartingales and the \gls{r-ast-rule} are introduced, the authors also propose automation techniques.
However, the authors of \cite{chatterjee_stochastic_2017} claim that their ``experimental results are basic`` and their computational methods are evaluated on only 3 examples, without having any available tool support.
For the benchmarks in Table~\ref{tab:bench:nast}, the outcomes of \amber{} and \amberlight{} coincide.
The reason for this is \gls{r-ast-rule}'s condition that the martingale expression has to have $c$-bounded differences.
This condition forces a suitable martingale expression to be bounded by a linear function, which is also the reason why \amber{} cannot certify the benchmark \emph{polynomial\_nast}.

\paragraph{Experimental Summary}
Our results from Tables~\ref{tab:bench:past}-\ref{tab:bench:nast} demonstrate that: 
\begin{itemize}
	\item \amber{} outperforms the state-of-the-art in automating \gls{past} certification for Prob-solvable loops (Table~\ref{tab:bench:past}).
	
	\item Complex probabilistic programs which are \gls{ast} and not \gls{past} as well as programs which are not \gls{ast} can automatically be certified as such by \amber{} (Tables~\ref{tab:bench:ast}, \ref{tab:bench:nast}).
	
	\item The relaxations of the proof rules introduced in Section~\ref{section:relaxations} are helpful in automating the termination analysis of probabilistic programs, as evidenced by the performance of \amber{} against \amberlight{} (Tables~\ref{tab:bench:past}-\ref{tab:bench:nast}). 
\end{itemize}

\begin{figure}[t!]
	\begin{minipage}{0.45\linewidth}
		\begin{table}[H]
			\tablesize
			\centering
			\bgroup
			\def\arraystretch{1.5}
			\begin{tabular}{lcc}
				\toprule
				Program & \amber{} & \amberlight{} \\
				\midrule
				fair\_in\_limit\_random\_walk & \benchna{} & \benchna{} \\
				\hdashline
				gambling & \benchsucc{} & \benchsucc{} \\
				\hdashline
				symmetric\_2d\_random\_walk & \benchfail{} & \benchfail{} \\
				\hdashline
				symmetric\_random\_walk\_constant\_1 & \benchsucc{} & \benchsucc{} \\
				\hdashline
				symmetric\_random\_walk\_constant\_2 & \benchsucc{} & \benchsucc{} \\
				\hdashline
				symmetric\_random\_walk\_exp\_1 & \benchsucc{} & \benchfail{} \\
				\hdashline
				symmetric\_random\_walk\_exp\_2 & \benchsucc{} & \benchfail{} \\
				\hdashline
				symmetric\_random\_walk\_linear\_1 & \benchsucc{} & \benchfail{} \\
				\hdashline
				symmetric\_random\_walk\_linear\_2 & \benchsucc{} & \benchsucc{} \\
				\hdashline
				symmetric\_random\_walk\_poly\_1 & \benchsucc{} & \benchfail{} \\
				\hdashline
				symmetric\_random\_walk\_poly\_2 & \benchsucc{} & \benchfail{} \\
				\midrule
				Total \benchsucc{} & 9 & 4 \\
				\bottomrule
			\end{tabular}
			\egroup
			\vspace{1em}
			\caption{11 programs which are \gls{ast} and not necessarily \gls{past}.}
			\label{tab:bench:ast}
		\end{table}
	\end{minipage}\hfill
	\begin{minipage}{0.45\linewidth}
		\vspace{5.65em}
		\begin{table}[H]
			\tablesize
			\centering
			\bgroup
			\def\arraystretch{1.5}
			\begin{tabular}{lcc}
				\toprule
				Program & \amber{} & \amberlight{} \\
				\midrule
				biased\_random\_walk\_nast\_1 & \benchsucc{} & \benchsucc{} \\
				\hdashline
				biased\_random\_walk\_nast\_2 & \benchsucc{} & \benchsucc{} \\	
				\hdashline
				biased\_random\_walk\_nast\_3 & \benchsucc{} & \benchsucc{} \\
				\hdashline
				biased\_random\_walk\_nast\_4 & \benchsucc{} & \benchsucc{} \\
				\hdashline
				binomial\_nast & \benchsucc{} & \benchsucc{} \\
				\hdashline
				polynomial\_nast & \benchfail{} & \benchfail{} \\
				\midrule
				Total \benchsucc{} & 5 & 5 \\
				\bottomrule
			\end{tabular}
			\egroup
			\vspace{1em}
			\caption{6 programs which are not \gls{ast}.}
			\label{tab:bench:nast}
		\end{table}
	\end{minipage}
\end{figure}
	
	\section{Related Work}
\label{section:related-work}

\paragraph*{Proof Rules for Probabilistic Termination}

Several proof rules have been proposed in the literature to provide sufficient conditions for the termination behavior of probabilistic programs. The work of~\cite{chakarov_probabilistic_2013} uses martingale theory to characterize \emph{positive almost sure termination (\gls{past})}.
In particular, the notion of a ranking supermartingale (\gls{rsm}) is introduced together with a proof rule (\gls{rsm-rule}) to certify \gls{past}, as discussed in Section~\ref{section:proof-rules-past}. The approach of~\cite{ferrer_fioriti_probabilistic_2015} extended this method  to include (demonic) non-determinism and continuous probability distributions, showing the completeness of the \gls{rsm-rule} for this program class.
The compositional approach proposed in~\cite{ferrer_fioriti_probabilistic_2015} was further strengthened
in~\cite{Huang19} to a sound approach using the notion of \emph{descent supermartingale map}. 
In~\cite{agrawal_lexicographic_2017}, the authors introduced \emph{lexicographic} \glspl{rsm}.

The \gls{sm-rule} discussed in Section~\ref{section:proof-rules-ast} was introduced in~\cite{mciver_new_2017}. 
It is worth mentioning that this proof rule is also applicable to non-deterministic probabilistic programs.
The work of~\cite{huang_new_2018} presented an independent proof rule based on supermartingales with lower bounds on conditional absolute differences. 
Both proof rules are based on supermartingales and can certify \gls{ast} for programs that are not necessarily \gls{past}.
The approach of~\cite{DBLP:conf/atva/TakisakaOUH18} examined martingale-based techniques for obtaining bounds on reachability probabilities --- and thus termination probabilities---  from an order-theoretic viewpoint.
The notions of \emph{nonnegative repulsing supermartingales} and \emph{$\gamma$-scaled submartingales}, accompanied by sound  and complete proof rules, have also been introduced. 
The \gls{r-ast-rule} from Section~\ref{section:proof-rules-r-ast} was proposed in~\cite{chatterjee_stochastic_2017} mainly for obtaining bounds on the probability of stochastic invariants.

An alternative approach is to exploit weakest precondition techniques for probabilistic programs, as presented in the seminal works~\cite{DBLP:journals/jcss/Kozen81,DBLP:journals/jcss/Kozen85} that can be used to certify \gls{ast}.
The work of~\cite{DBLP:series/mcs/McIverM05} extended this approach to programs with non-determinism and provided several proof rules for termination.
These techniques are purely syntax-based.
In~\cite{DBLP:journals/jacm/KaminskiKMO18} a weakest precondition calculus for obtaining bounds on expected termination times was proposed.
This calculus comes with proof rules to reason about loops.

\paragraph*{Automation of Martingale Techniques}
The work of~\cite{chakarov_probabilistic_2013} proposed an automated procedure --- by using Farkas' lemma --- to synthesize \emph{linear} (super)martingales for probabilistic programs with linear variable updates.
This technique was considered in our experimental evaluation, cf.\ Section~\ref{section:implementation}. 
The algorithmic construction of supermartingales was extended to treat (demonic) non-determinism in~\cite{chatterjee_algorithmic_2018} and to polynomial supermartingales in~\cite{chatterjee_termination_2016} using semi-definite programming. 
The recent work of~\cite{ChenH20} uses $\omega$-regular decomposition to certify \gls{ast}. 
They exploit so-called \emph{localized} ranking supermartingales, which can be synthesized efficiently but must be linear.

\paragraph*{Other Approaches}
Abstract interpretation is used in~\cite{DBLP:conf/sas/Monniaux01}  to prove the probabilistic termination of programs for which the probability of taking a loop $k$ times decreases at least exponentially with $k$. 
In~\cite{esparza_proving_2012}, a sound and complete procedure deciding \gls{ast} is given for probabilistic programs with a finite number of reachable states from any initial state.
The work of \cite{ngo_bounded_2018} gave an algorithmic approach based on potential functions for computing bounds on the expected resource consumption of probabilistic programs.  
In~\cite{DBLP:conf/tacas/LengalLMR17}, model checking is exploited to automatically verify whether a parameterized family of probabilistic concurrent systems is \gls{ast}.

Finally, the class of Prob-solvable loops considered in this paper extends~\cite{bartocci_automatic_2019} to a wider class of loops. While~\cite{bartocci_automatic_2019} focused on computing statistical higher-order moments, our work addresses the termination behavior of probabilistic programs.
The related approach of~\cite{GieslGH19} computes exact expected runtimes of constant probability programs and provides a decision procedure for \gls{ast} and \gls{past} for such programs. 
Our programming model strictly generalizes the constant probability programs of~\cite{GieslGH19}, by supporting polynomial loop guards, updates and martingale expressions.

	
	\section{Conclusion}
\label{section:conclusion}

This paper reported on the automation of termination analysis of probabilistic while-programs whose guards and expressions are polynomial expressions. 
To this end, we introduced mild relaxations of existing proof rules for \gls{ast}, \gls{past}, and their negations, by requiring their sufficient conditions to hold only eventually.
The key to our approach is that the structural constraints of Prob-solvable loops allow for automatically computing almost sure asymptotic bounds on polynomials over program variables.
Prob-solvable loops cover a vast set of complex and relevant probabilistic processes including random walks and dynamic Bayesian networks~\cite{bartocci2020analysis}.
Only two out of 50 benchmarks in \cite{chakarov_probabilistic_2013,ngo_bounded_2018} are outside the scope of Prob-solvable loops regarding \gls{past} certification.
The almost sure asymptotic bounds were used to formalize algorithmic approaches for proving \gls{ast}, \gls{past}, and their negations.
Moreover, for Prob-solvable loops four different proof rules from the literature uniformly come together in our work.

Our approach is implemented in the software tool \amber{} (\href{https://github.com/probing-lab/amber}{github.com/probing-lab/amber}), offering a fully automated approach to probabilistic termination.
Our experimental results show that our relaxed proof rules enable proving probabilistic (non-)\allowbreak{}termination of more programs than could be treated before.
A comparison to the state-of-art in automated analysis of probabilistic termination reveals that \amber{} significantly outperforms related approaches.
To the best of our knowledge, \amber{} is the first tool to automate \gls{ast}, \gls{past}, non-\gls{ast} and non-\gls{past} in a single tool-chain.

There are several directions for future work.
These include extensions to Prob-solvable loops such as symbolic distributions, more complex control flow, and non-determinism.
We will also consider program transformations that translate programs into our format.
Extensions of the \gls{sm-rule} algorithm with non-constant probability and decrease functions are also in our interest.

	%
	%
	%
	\clearpage
	\bibliographystyle{splncs04}
	\bibliography{paper}
	
	\vfill
	
	{\small\medskip\noindent{\bf Open Access} This chapter is licensed under the terms of the Creative Commons\break Attribution 4.0 International License (\url{http://creativecommons.org/licenses/by/4.0/}), which permits use, sharing, adaptation, distribution and reproduction in any medium or format, as long as you give appropriate credit to the original author(s) and the source, provide a link to the Creative Commons license and indicate if changes were made.}
	
	{\small \spaceskip .28em plus .1em minus .1em The images or other third party material in this chapter are included in the chapter's Creative Commons license, unless indicated otherwise in a credit line to the material.~If material is not included in the chapter's Creative Commons license and your intended\break use is not permitted by statutory regulation or exceeds the permitted use, you will need to obtain permission directly from the copyright holder.}
	
	\medskip\noindent\includegraphics{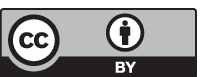}
\end{document}